\documentclass[11pt,english,onecolumn]{IEEEtran}
\usepackage[T1]{fontenc}
\usepackage[latin1]{inputenc}
\usepackage{amsmath}
\usepackage{setspace}
\usepackage{amssymb}

\makeatletter
\newtheorem{definitn}{Definition}
\newtheorem{thm}{Theorem}
\newtheorem{remrk}{Remark}
\newtheorem{lemma}{Lemma}
\newtheorem{prop}{Proposition}

\usepackage{cite}
\usepackage{bm}
\usepackage{bbm}
\usepackage{mathrsfs}

\makeatother

\usepackage{babel}

\begin{document}

\title{Weighted Superimposed Codes and Constrained Integer Compressed Sensing}

\author{Wei Dai and Olgica Milenkovic \\
{\small Dept.\ of Electrical and Computer Engineering }{\normalsize{} }\\
{\normalsize{} }{\small University of Illinois, Urbana-Champaign}{\normalsize{} }}

\maketitle
\begin{abstract}
We introduce a new family of codes, termed weighted superimposed codes
(WSCs). This family generalizes the class of Euclidean superimposed
codes (ESCs), used in multiuser identification systems. WSCs allow
for discriminating all bounded, integer-valued linear combinations
of real-valued codewords that satisfy prescribed norm and non-negativity
constraints. By design, WSCs are inherently noise tolerant. Therefore,
these codes can be seen as special instances of robust compressed
sensing schemes. The main results of the paper are lower and upper
bounds on the largest achievable code rates of several classes of
WSCs. These bounds suggest that with the codeword and weighting vector
constraints at hand, one can improve the code rates achievable by
standard compressive sensing. 
\end{abstract}
\renewcommand{\thefootnote}{\fnsymbol{footnote}} \footnotetext[1]{This work is supported by the NSF Grant CCF 0644427, the NSF Career Award, and the DARPA Young Faculty Award of the second author. Parts of the results were presented at the CCIS'2008 ant ITW'2008 conferences.} \renewcommand{\thefootnote}{\arabic{footnote}} \setcounter{footnote}{0}

\section{Introduction}

Superimposed codes (SCs) and designs were introduced by Kautz and
Singleton \cite{Kautz_IT1964_Superimposed}, for the purpose of studying
database retrieval and group testing problems. In their original formulation,
superimposed designs were defined as arrays of binary codewords with
the property that bitwise OR functions of all sufficiently small collections
of codewords are distinguishable. Superimposed designs can therefore
be viewed as binary {}``parity-check'' matrices for which syndromes
represent bitwise OR, rather than XOR, functions of selected sets
of columns.

The notion of binary superimposed codes was further generalized by
prescribing a distance constraint on the OR evaluations of subsets
of columns, and by extending the fields in which the codeword symbols
lie \cite{Ericson_IT1988_SuperImposed_Codes_Rn}. In the latter case,
Ericson and Györfi introduced Euclidean superimposed codes (ESCs),
for which the symbol field is $\mathbb{R}$, for which the OR function
is replaced by real addition, and for which all sums of less than
$K$ codewords are required to have pairwise Euclidean distance at
least $d$. The best known upper bound on the size of Euclidean superimposed
codes was derived by Füredi and Ruszinko \cite{Furedi_IT1999_ub_rate_superimposed_codes},
who used a combination of sphere packing arguments and probabilistic
concentration formulas to prove their result.

On the other hand, compressed sensing (CS) is a new sampling method
usually applied to $K$-\emph{sparse signals}, i.e.\ signals embedded
in an $N$-dimensional space that can be represented by only $K\ll N$
significant coefficients~\cite{Donoho_IT2006_CompressedSensing,Candes_Tao_IT2006_Robust_Uncertainty_Principles,Candes_Tao_ApplMath2006_Stable_Signal_Recovery}.
Alternatively, when the signal is projected onto a properly chosen
basis of the transform space, its accurate representation relies only
on a small number of coefficients. Encoding of a $K$-sparse discrete-time
signal $\textbf{x}$ of dimension $N$ is accomplished by computing
a measurement vector $\textbf{y}$ that consists of $m\ll N$ linear
projections, i.e.\ $\textbf{y}=\Phi\textbf{x}$. Here, $\Phi$ represents
an $m\times N$ matrix, usually over the field of real numbers. Consequently,
the measured vector represents a linear combination of columns of
the matrix $\Phi$, with weights prescribed by the nonzero entries
of the vector $\textbf{x}$. Although the reconstruction of the signal
$\textbf{x}\in\mathbb{R}^{N}$ from the (possibly noisy) projections
is an ill-posed problem, the prior knowledge of signal sparsity allows
for accurate recovery of $\textbf{x}$.

The connection between error-correcting coding theory and compressed
sensing was investigated by Cand{è}s and Tao in \cite{Candes_Tao_IT2006_Near_Optimal_Signal_Recovery},
and remarked upon in \cite{Cormode_2003}. In the former work, the
authors studied random codes over the real numbers, the noisy observations
of which can be decoded using linear programming techniques. As with
the case of compressed sensing, the performance guarantees of this
coding scheme are probabilistic, and the $K$-sparse signal is assumed
to lie in $\mathbb{R}^{N}$.

We propose to study a new class of codes, termed \emph{weighted superimposed
codes} (WSCs), which provide a link between SCs and CS matrices. As
with the case of the former two entities, WSCs are defined over the
field of real numbers. But unlike ESCs, for which the sparse signal
$\textbf{x}$ consists of zeros and ones only, and unlike CS, for
which $\textbf{x}$ is assumed to belong to $\mathbb{R}^{N}$, in
WSCs the vector $\textbf{x}$ is drawn from $B^{N}$, where $B$ denotes
a \emph{bounded, symmetric set of integers}. The motivation for studying
WSCs comes from the fact that in many applications, the alphabet of
the sparse signal can be modeled as a finite set of integers.

Codewords from the family of WSCs can be designed to obey prescribed
norm and non-negativity constraints. The restriction of the weighting
coefficients to a bounded set of integers ensures reconstruction robustness
in the presence of noise - i.e., all weighted sums of at most $K$
codewords can be chosen at {}``minimum distance'' $d$ from each
other. This minimum distance property provides deterministic performance
guarantees, which CS techniques usually lack. Another benefit of the
input alphabet restriction is the potential to reduce the decoding
complexity compared to that of CS reconstruction techniques. This
research problem was addressed by the authors in \cite{Dai2008_ITW_Weighted_Euclidean_Superimposed_Code,Dai2008_ITA_constrained_CS,Dai2008_CISS_sparse_superimposed_codes},
but is beyond the scope of this paper.

The central problem of this paper is to characterize the rate region
for which a WSC with certain parameters exists. The main results of
this work include generalized sphere packing upper bounds and random
coding lower bounds on the rates of several WSC families. The upper
and lower bounds differ only by a constant, and therefore imply that
the superposition constraints are ensured whenever $m=O\left(K\log N/\log K\right)$.
In the language of CS theory, this result suggests that the number
of required signal measurements is less than the standard $O\left(K\log\left(N/K\right)\right)$,
required for discriminating real-valued linear combinations of codewords.
This reduction in the required number of measurements (codelength)
can be seen as a result of restricting the input alphabet of the sparse
signal.

The paper is organized as follows. Section~\ref{sec:definitions}
introduces the relevant terminology and definitions. Section~\ref{sec:cardinality-WSCs}
contains the main results of the paper -- upper and lower bounds on
the size of WSCs. The proofs of the rate bounds are presented in Sections~\ref{sec:Upper-Bound}
and~\ref{sec:lb-ngL1WSC}. Concluding remarks are given in Section~\ref{sec:Conclusion}.

\section{Motivating Applications}

We describe next two applications - one arising in wireless communication,
the other in bioengineering - motivating the study of WSCs.

\emph{The adder channel and signature codes}: One common application
of ESCs is for signaling over multi-access channels. For a given set
of $k\leq K$ active users in the channel, the input to the receiver
$\textbf{y}$ equals the sum of the signals (signatures) $\textbf{x}_{i_{j}},\, j=1,\ldots,k,$
of the $k$ active users, i.e.\ $\mathbf{y}=\sum_{j=1}^{k}\mathbf{x}_{i_{j}}$.
The signatures are only used for identification purposes, and in order
to minimize energy consumption, all users are assigned unit energy
\cite{Ericson_IT1988_SuperImposed_Codes_Rn,Furedi_IT1999_ub_rate_superimposed_codes}.
Now, consider the case that in addition to identifying their presence,
active users also have to convey some limited information to the receiver
by adapting their transmission power. The received signal can in this
case be represented by a weighted sum of the signatures of active
users, i.e., $\mathbf{y}=\sum_{j=1}^{k}\sqrt{p_{i_{j}}}\mathbf{x}_{i_{j}}.$
The codebook used in this scheme represents a special form of WSC,
termed Weighted Euclidean Superimposed Codes (WESCs); these codes
are formally defined in Section \ref{sec:definitions}.

\emph{Compressive sensing microarrays}: A microarray is a bioengineering
device used for measuring the level of certain molecules, such as
RNA (ribonucleic acid) sequences, representing the joint expression
profile of thousands of genes. A microarray consist of thousands of
microscopic spots of DNA sequences, called probes. The complementary
DNA (cDNA) sequences of RNA molecules being measured are labeled with
fluorescent tags, and such units are termed targets. If a target sequence
has a significant homology with a probe sequence on the microarray,
the target cDNA and probe DNA molecules will bind or {}``hybridize''
so as to form a stable structure. As a result, upon exposure to laser
light of the appropriate wavelength, the microarray spots with large
hybridization activity will be illuminated. The specific illumination
pattern and intensities of microarray spots can be used to infer the
concentration of RNA molecules. In traditional microarray design,
each spot of probes is a unique identifier of only one target molecule.
In our recent work \cite{Sheikh2007_CAMSAP_microarray,Dai2008_BIBM_microarray},
we proposed the concept of \emph{compressive sensing microarrays}
(CSM), for which each probe has the potential to hybridize with several
different targets. It uses the observation that, although the number
of potential target RNA types is large, not all of them are expected
to be present in a significant concentration at all observed times.

Mathematically, a microarray is represented by a measurement matrix,
with an entry in the $i^{\mathrm{th}}$ row and the $j^{\mathrm{th}}$
column corresponding to the hybridization probability between the
$i^{\mathrm{th}}$ probe and the $j^{\mathrm{th}}$ target. In this
case, all the entries in the measurement matrix are nonnegative real
numbers, and all the columns of the measurement matrix are expected
to have $l_{1}$-norms equal to one. In microarray experiments, the
input vector $\mathbf{x}$ has entries that correspond to integer
multiples of the smallest detectable concentration of target cDNA
molecules. Since the number of different target cDNA types in a typical
test sample is small compared to the number of all potential types,
one can assume that the vector $\mathbf{x}$ is sparse. Furthermore,
the number of RNA molecules in a cell at any point in time is upper
bounded due to energy constraints, and due to intracellular space
limitations. Hence, the integer-valued entries of $\mathbf{x}$ are
assumed to have bounded magnitudes and to be relatively small compared
to the number of different RNA types. With the above considerations,
the measurement matrix of a CSM can be described by nonnegative $l_{1}$-WSCs,
formally defined in Section \ref{sec:definitions}.

\section{\label{sec:definitions}Definitions and Terminology}

Throughout the paper, we use the following notation and definitions.

A code $\mathcal{C}$ is a finite set of $N$ codewords (vectors)
$\mathbf{v}_{i}\in\mathbb{R}^{m\times1}$, $i=1,2,\cdots,N$. The
code $\mathcal{C}$ is specified by its \emph{codeword matrix (codebook)}
$\mathbf{C}\in\mathbb{R}^{m\times N}$, obtained by arranging the
codewords in columns of the matrix.

For two given positive integers, $t$ and $K$, let \[
B_{t}=\left[-t,t\right]=\left\{ -t,-t+1,\cdots,t-1,t\right\} \subset\mathbb{Z}\]
 be a symmetric, bounded set of integers, and let \[
\mathcal{B}_{K}=\left\{ \mathbf{b}\in B_{t}^{N}:\;\left\Vert \mathbf{b}\right\Vert _{0}\le K\right\} \]
 denote the $l_{0}$ ball of radius $K$, with $\left\Vert \mathbf{b}\right\Vert _{0}$
representing the number of nonzero components in the vector $\mathbf{b}$
(i.e., the support size of the vector). We formally define WESCs as
follows.

\vspace{0.05in}

\begin{definitn}
\label{def:WESC}A code $\mathcal{C}$ is said to be a WESC with parameters
$\left(N,m,K,d,\eta,B_{t}\right)$ for some $d\in\left(0,\eta\right)$,
if 
\begin{enumerate}
\item $\mathbf{C}\in\mathbb{R}^{m\times N}$, 
\item $\left\Vert \mathbf{v}_{i}\right\Vert _{2}=\eta,$ for all $i=1,\cdots,N$,
and, 
\item if the following minimum distance property holds: \[
d_{E}\left(\mathcal{C},K,B_{t}\right):=\underset{\mathbf{b}_{1}\ne\mathbf{b}_{2}}{\min}\left\Vert \mathbf{C}\mathbf{b}_{1}-\mathbf{C}\mathbf{b}_{2}\right\Vert _{2}\ge d\]
 for all $\mathbf{b}_{1},\mathbf{b}_{2}\in\mathcal{B}_{K}$. 
\end{enumerate}
\end{definitn}
\vspace{0.05in}

Henceforth, we focus our attention on WESCs with $\eta=1$, and denote
the set of parameters of interest by $\left(N,m,K,d,B_{t}\right)$.

The definition above can be extended to hold for other normed spaces.

\vspace{0.05in}

\begin{definitn}
\label{def:LpWSC}A code $\mathcal{C}$ is said to be an $l_{p}$-WSC
with parameters $\left(N,m,K,d,B_{t}\right)$ if 
\begin{enumerate}
\item $\mathbf{C}\in\mathbb{R}^{m\times N}$, 
\item $\left\Vert \mathbf{v}_{i}\right\Vert _{l_{p}}=1$, for all $i=1,\cdots,N$,
and, 
\item if the following minimum distance property holds: \[
d_{p}\left(\mathcal{C},K,B_{t}\right):=\underset{\mathbf{b}_{1}\ne\mathbf{b}_{2}}{\min}\left\Vert \mathbf{C}\mathbf{b}_{1}-\mathbf{C}\mathbf{b}_{2}\right\Vert _{l_{p}}\ge d\]
 for all $\mathbf{b}_{1},\mathbf{b}_{2}\in\mathcal{B}_{K}$. 
\end{enumerate}
\end{definitn}
\vspace{0.05in}

Note that specializing $p=2$ reproduces the definition of a WESC.

Motivated by the practical applications described in the previous
section, we also define the class of \emph{nonnegative} $l_{p}$-WSC.
\vspace{0.05in}

\begin{definitn}
\label{def:mgLpWSC}A code $\mathcal{C}$ is said to be a nonnegative
$l_{p}$-WSC with parameters $\left(N,m,K,d,B_{t}\right)$ if it is
an $l_{p}$-WSC such that all entries of $\mathbf{C}$ are nonnegative. 
\end{definitn}
\vspace{0.05in}

Given the parameters $m,$ $K,$ $d$ and $B_{t}$, let $N\left(m,K,d,B_{t}\right)$
denote the maximum size of a WSC, \[
N\left(m,K,d,B_{t}\right):=\max\left\{ N:\;\mathcal{C}\left(N,m,K,d,B_{t}\right)\ne\phi\right\} .\]
 The \emph{asymptotic code exponent} is defined as \[
R\left(K,d,B_{t}\right):=\underset{m\rightarrow\infty}{\lim\;\sup}\;\frac{\log N\left(m,K,d,B_{t}\right)}{m}.\]

We are interested in quantifying the asymptotic code exponent of WSCs,
and in particular, WESCs and nonnegative WSCs with $p=1$. Results
pertaining to these classes of codes are summarized in the next section.

\section{\label{sec:cardinality-WSCs}On the Cardinality of WSC Families}

The central problem of this paper is to determine the existence of
a superimposed code with certain parameters. In \cite{Ericson_IT1988_SuperImposed_Codes_Rn,Furedi_IT1999_ub_rate_superimposed_codes},
it was shown that for ESCs, for which the codeword alphabet $B_{t}$
is replaced by the asymmetric set $\left\{ 0,1\right\} $, one has
\[
\frac{\log K}{4K}\left(1+o_{d}\left(1\right)\right)\le R\left(K,d,\left\{ 1\right\} \right)\le\frac{\log K}{2K}\left(1+o_{d}\left(1\right)\right),\]
 where $o_{d}(1)$ converges to zero as $K\to\infty$.

The main result of the paper is the upper and lower bounds on the
asymptotic code exponents of several WSC families. For WESCs, introducing
weighting coefficients larger than one does not change the asymptotic
order of the code exponent.

\vspace{0.05in}

\begin{thm}
\label{thm:WESC-code-exponent} Let $t$ be a fixed parameter. For
sufficiently large $K$, the asymptotic code exponent of WESCs can
be bounded as \begin{equation}
\frac{\log K}{4K}\left(1+o\left(1\right)\right)\le R\left(K,d,B_{t}\right)\le\frac{\log K}{2K}\left(1+o_{t,d}\left(1\right)\right)\label{eq:WESC-code-exponent}\end{equation}
 where $o\left(1\right)\rightarrow0$ and $o_{t,d}\left(1\right)\rightarrow0$
as $K\rightarrow\infty$. The exact expressions of the $o\left(1\right)$
and $o_{t,d}\left(1\right)$ terms are given in Equations (\ref{eq:o_lb_Euclidean_main})
and (\ref{eq:o_ub_Euclidean}), respectively. 
\end{thm}
\vspace{0.05in}

\begin{remrk}
The derivations leading to the expressions in Theorem~\ref{thm:WESC-code-exponent}
show that one can also bound the code exponent in a non-asymptotic
regime. Unfortunately, those expressions are too complicated for practical
use. Nevertheless, this observation implies that the results pertaining
to WESC are applicable for the same parameter regions as those arising
in the context of CS theory. 
\end{remrk}
\vspace{0.05in}

\begin{remrk}
The parameter $t$ can also be allowed to increase with $K$. For
WESCs, the value of $t$ does not affect the lower bound on the asymptotic
code exponent, while the upper bound is valid as long as $t=o\left(K\right)$. 
\end{remrk}
\vspace{0.05in}

For clarity of exposition, the proof of the lower bound is postponed
to Section \ref{sec:lb-Euclidean}, while the proof of the upper bound,
along with the proofs of the upper bounds for other WSC families,
are presented in Section \ref{sec:Upper-Bound}. We briefly sketch
the main steps of the proofs in the discussion that follows.

The proof of the upper bound is based on the sphere packing argument.
The classical sphere packing argument is valid for all WSC families
discussed in this paper. The leading term of the resulting upper bound
is $\left(\log K\right)/K$. This result can be improved when restricting
one's attention to the Euclidean norm. The key idea is to show that
most points of the form $\mathbf{Cb}$ lie in a ball of radius significantly
smaller than the one derived by the classic sphere packing argument.
The leading term of the upper bound can in this case be improved from
$\left(\log K\right)/K$ to $\left(\log K\right)/\left(2K\right)$.

The lower bound in Theorem \ref{thm:WESC-code-exponent} is proved
by random coding arguments. We first randomly generate a family of
WESCs from the Gaussian ensemble, with the code rates satisfying \[
\underset{\left(m,N\right)\rightarrow\infty}{\lim}\frac{\log N}{m}<\frac{\log K}{4K}\left(1+o\left(1\right)\right).\]
 Then we prove that these randomly generated codebooks satisfy \[
d_{E}\left(\mathcal{C},K\right)\ge d\]
 with high probability. This fact implies that the asymptotic code
exponent \begin{align*}
R\left(K,d,B_{t}\right) & =\underset{m\rightarrow\infty}{\lim\;\sup}\;\frac{\log N\left(m,K,d,B_{t}\right)}{m}\\
 & \ge\frac{\log K}{4K}\left(1+o\left(1\right)\right).\end{align*}

We also analyze two more classes of WSCs: the class of general $l_{1}$-WSCs
and the family of nonnegative $l_{1}$-WSCs. The characterization
of the asymptotic code rates of these codes is given in Theorems~\ref{thm:L1WSC-code-exponent}
and~\ref{thm:ngL1WSC-code-exponent}, respectively.

\vspace{0.05in}

\begin{thm}
\label{thm:L1WSC-code-exponent} For a fixed value of the parameter
$t$ and sufficiently large $K$, the asymptotic code exponent of
$l_{1}$-WSCs is bounded as \begin{equation}
\frac{\log K}{4K}\left(1+o\left(1\right)\right)\le R\left(K,d,B_{t}\right)\le\frac{\log K}{K}\left(1+o_{t,d}\left(1\right)\right),\label{eq:L1WSC-code-exponent}\end{equation}
 where the expressions for $o\left(1\right)$ and $o_{t,d}\left(1\right)$
are given in Equations~(\ref{eq:o_lb_L1_WSC-2}) and~(\ref{eq:o_ub_WSCs}),
respectively. 
\end{thm}
\begin{proof}
The lower bound is proved in Section~\ref{sec:lb-L1WSC}, while the
upper bound is proved in Section~\ref{sec:Upper-Bound}. 
\end{proof}
\vspace{0.05in}

\begin{thm}
\label{thm:ngL1WSC-code-exponent} For a fixed value of the parameter
$t$ and sufficiently large $K$, the asymptotic code exponent of
\emph{nonnegative} $l_{1}$-WSCs is bounded as \begin{equation}
\frac{\log K}{4K}\left(1+o_{t}\left(1\right)\right)\le R\left(K,d,B_{t}\right)\le\frac{\log K}{K}\left(1+o_{t,d}\left(1\right)\right),\label{eq:ngL1WSC-code-exponent}\end{equation}
 where the expressions for $o_{t}\left(1\right)$ and $o_{t,d}\left(1\right)$
are given by Equations~(\ref{eq:o_lb_ngL1WSC-2}) and~(\ref{eq:o_ub_WSCs}),
respectively. 
\end{thm}
\begin{proof}
The lower and upper bounds are proved in Sections~\ref{sec:lb-ngL1WSC}
and~\ref{sec:Upper-Bound}, respectively. 
\end{proof}
\vspace{0.05in}

\begin{remrk}
The upper bounds in Equations~(\ref{eq:L1WSC-code-exponent}) and~(\ref{eq:ngL1WSC-code-exponent})
also hold if one allows $t$ to grow with $K$, so that $t=o\left(K\right)$.
The lower bound in (\ref{eq:L1WSC-code-exponent}) for general $l_{1}$-WSCs
does not depend on the value of $t$. However, the lower bound~(\ref{eq:ngL1WSC-code-exponent})
for nonnegative $l_{1}$-WSCs requires that $t=o\left(K^{1/3}\right)$
(see Equation (\ref{eq:o_lb_ngL1WSC-2}) for details). This difference
in the convergence regime of the two $l_{1}$-WSCs is a consequence
of the use of different proof techniques. For the proof of the rate
regime of general $l_{1}$-WSCs, Gaussian codebooks were used. On
the other hand, for nonnegative $l_{1}$-WSCs, the analysis is complicated
by the fact that one has to analyze linear combinations of nonnegative
random variables. To overcome this difficulty, we used the Central
Limit Theorem and Berry-Essen type of distribution approximations~\cite{Book_Feller_II}.
The obtained results depend on the value of $t$. 
\end{remrk}
\vspace{0.05in}

\begin{remrk}
The upper bound for WESCs is roughly one half of the corresponding
bound for $l_{1}$-WSCs. This improvement in the code exponent of
WESCs rests on the fact that the $l_{2}$-norm of a vector can be
expressed as an inner product, i.e.\
$\left\Vert \mathbf{v}\right\Vert _{2}^{2}=\mathbf{v}^{\dagger}\mathbf{v}$
(in other words, $l_{2}$ is a Hilbert space). Other normed spaces
considered in the paper lack this property, and at the present, we
are not able to improve the upper bounds for $l_{p}$-WSCs with $p\ne2$. 
\end{remrk}
\vspace{0.05in}

\section{\label{sec:Upper-Bound}Proof of the Upper Bounds by Sphere packing
Arguments}

It is straightforward to apply the sphere packing argument to upper
bound the code exponents of WSCs. Regard an $l_{p}$-WSC with arbitrary
$p\in\mathbb{Z}^{+}$. The superposition $\mathbf{C}\mathbf{b}$ satisfies
\[
\left\Vert \mathbf{Cb}\right\Vert _{p}\le\sum_{j=1}^{\left\Vert \mathbf{b}\right\Vert _{0}}\left\Vert \mathbf{v}_{i_{j}}b_{i_{j}}\right\Vert _{p}\le Kt\]
 for all $\mathbf{b}$ such that $\left\Vert \mathbf{b}\right\Vert _{0}\le K$,
where the $b_{i_{j}}$s, $1\le j\le\left\Vert \mathbf{b}\right\Vert _{0}\le K$,
denote the nonzero entries of $\mathbf{b}$. Note that the $l_{p}$
distance of any two superpositions is required to be at least $d$.
The size of the $l_{p}$-WSC codebook, $N$, satisfies the sphere
packing bound \begin{equation}
\sum_{k=1}^{K}{N \choose k}\left(2t\right)^{k}\le\left(\frac{tK+\frac{d}{2}}{\frac{d}{2}}\right)^{m}.\label{eq:sphere-packing-bd}\end{equation}
 A simple algebraic manipulation of the above equation shows \[
\sum_{k=1}^{K}{N \choose k}\left(2t\right)^{k}\ge{N \choose K}\left(2t\right)^{K}\ge\left(\frac{N-K}{K}\right)^{K}\left(2t\right)^{K},\]
 so that one has \begin{align*}
\frac{\log N}{m} & \le\frac{1}{K}\log\left(1+\frac{2tK}{d}\right)-\frac{\log\left(2t\right)}{m}-\frac{\log\left(\frac{1}{K}-\frac{1}{N}\right)}{m}\\
 & =\frac{\log K}{K}+\frac{1}{K}\log\left(\frac{2t}{d}+\frac{1}{K}\right)\\
 & \quad-\frac{\log\left(2t\right)}{m}-\frac{\log\left(\frac{1}{K}-\frac{1}{N}\right)}{m}.\end{align*}
 The asymptotic code exponent is therefore upper bounded by \begin{equation}
\frac{\log K}{K}\left(1+o_{t,d}\left(1\right)\right),\label{eq:sphere-packing-bd-rate}\end{equation}
 where \begin{align}
o_{t,d}\left(1\right) & =\frac{\log\left(\frac{2t}{d}+\frac{1}{K}\right)}{\log K}\overset{K\rightarrow\infty}{\longrightarrow}0\label{eq:o_ub_WSCs}\end{align}
 if $t=o\left(K\right)$.

This sphere packing bound can be significantly improved when considering
the Euclidean norm. The result is an upper bound with the leading
term $\left(\log K\right)/\left(2K\right)$. The proof is a generalization
of the ideas used by Füredi and Ruszinko in~\cite{Furedi_IT1999_ub_rate_superimposed_codes}:
most points of the form $\mathbf{Cb}$ lie in a ball with radius smaller
than $\sqrt{\frac{K}{3}}\left(t+1\right)$, and therefore the right
hand side of the classic sphere packing bound (\ref{eq:sphere-packing-bd-rate})
can be reduced by a factor of two.

To proceed, we assign to every $\mathbf{b}\in\mathcal{B}_{K}$ the
probability \[
\frac{1}{\left|\mathcal{B}_{K}\right|}=\frac{1}{\sum_{k=1}^{K}{N \choose k}\left(2t+1\right)^{k}}.\]
 For a given codeword matrix $\mathbf{C}$, define a random variable
\[
\xi=\left\Vert \mathbf{C}\mathbf{b}\right\Vert _{2}.\]
 We shall upper bound the probability $\Pr\left\{ \xi\ge\lambda\mu\right\} $,
for arbitrary $\lambda,\mu\in\mathbb{R}^{+}$, via Markov's inequality
\[
\Pr\left(\xi\ge\lambda\mu\right)\le\frac{\mathrm{E}\left[\xi\right]}{\lambda\mu}\le\frac{\sqrt{\mathrm{E}\left[\xi^{2}\right]}}{\lambda\mu}.\]

We calculate $\mathrm{E}\left[\xi^{2}\right]$ as follows. For a given
vector $\mathbf{b}$, let $I\subset\left[1,N\right]$ be its support
set - i.e., the set of indices for which the entries of $\mathbf{b}$
are nonzero. Let $\mathbf{b}_{I}$ be the vector composed of the nonzero
entries of $\mathbf{b}$. Furthermore, define \[
B_{t,k}=\left(B_{t}\backslash\left\{ 0\right\} \right)^{k}.\]
 Then,

\[
E\left[\xi^{2}\right]=\frac{1}{\left|\mathcal{B}_{K}\right|}\sum_{k=1}^{K}\;\sum_{\left|I\right|=k}\;\sum_{\mathbf{b}_{I}\in B_{t,k}}\;\left\Vert \sum_{j=1}^{k}b_{i_{j}}\mathbf{v}_{i_{j}}\right\Vert _{2}^{2},\]
 where $i_{j}\in I$, $j=1,\cdots,k$. Note that \begin{align*}
 & \sum_{\left|I\right|=k}\;\sum_{\mathbf{b}_{I}\in B_{t,k}}\;\left\Vert \sum_{j=1}^{k}b_{i_{j}}\mathbf{v}_{i_{j}}\right\Vert _{2}^{2}\\
 & =\sum_{\left|I\right|=k}\;\sum_{\mathbf{b}_{I}\in B_{t,k}}\left(\sum_{j=1}^{k}b_{i_{j}}^{2}+\sum_{1\le l\ne j\le k}b_{i_{j}}b_{i_{l}}\mathbf{v}_{i_{j}}^{\dagger}\mathbf{v}_{i_{l}}\right)\\
 & =\underbrace{\sum_{\left|I\right|=k}\sum_{\mathbf{b}_{I}\in B_{t,k}}\sum_{j=1}^{k}b_{i_{j}}^{2}}_{\left(*\right)}+\underbrace{\sum_{\left|I\right|=k}\sum_{\mathbf{b}_{I}\in B_{t,k}}\sum_{1\le l\ne j\le k}b_{i_{j}}b_{i_{l}}\mathbf{v}_{i_{j}}^{\dagger}\mathbf{v}_{i_{l}}}_{\left(**\right)}.\end{align*}
 It is straightforward to evaluate the two sums in the above expression
in closed form: \begin{align*}
\left(*\right) & ={N \choose k}\sum_{\mathbf{b}_{I}\in B_{t,k}}\;\sum_{j=1}^{k}b_{i_{j}}^{2}={N \choose k}k\sum_{\mathbf{b}_{I}\in B_{t,k}}b_{i_{1}}^{2}\\
 & ={N \choose k}k\left(2t\right)^{k-1}\sum_{b_{i_{1}}\in B_{t,1}}b_{i_{1}}^{2}\\
 & ={N \choose k}k\left(2t\right)^{k-1}\frac{t\left(t+1\right)\left(2t+1\right)}{3};\end{align*}
 and \begin{align*}
\left(**\right) & ={N \choose k}\sum_{1\le l\ne j\le k}\;\sum_{\mathbf{b}_{I}\in B_{t,k}}b_{i_{l}}b_{i_{j}}\mathbf{v}_{i_{l}}^{\dagger}\mathbf{v}_{i_{j}}\\
 & ={N \choose k}\left(2t\right)^{k-2}\sum_{1\le i\ne j\le k}\;\sum_{b_{i_{l}},b_{i_{j}}\in B_{t,1}}b_{i_{l}}b_{i_{j}}\mathbf{v}_{i_{l}}^{\dagger}\mathbf{v}_{i_{j}}\\
 & =0,\end{align*}
 where the last equality follows from the observation that \begin{align*}
 & \sum_{b_{i_{l}}\in B_{t,1},b_{i_{j}}\in B_{t,1}}b_{i_{l}}b_{i_{j}}\mathbf{v}_{i_{l}}^{\dagger}\mathbf{v}_{i_{j}}\\
 & =\sum_{b_{i_{l}}>0,b_{i_{j}}\in B_{t,1}}b_{i_{l}}b_{i_{j}}\mathbf{v}_{i_{l}}^{\dagger}\mathbf{v}_{i_{j}}\\
 & \quad+\sum_{b_{i_{l}}>0,b_{i_{j}}\in B_{t,1}}\left(-b_{i_{l}}\right)b_{i_{j}}\mathbf{v}_{i_{l}}^{\dagger}\mathbf{v}_{i_{j}}\\
 & =0.\end{align*}
 Consequently, one has \begin{align*}
 & \sum_{\left|I\right|=k}\;\sum_{\mathbf{b}_{I}\in B_{t,k}}\;\left\Vert \sum_{j=1}^{k}b_{i_{j}}\mathbf{v}_{i_{j}}\right\Vert _{2}^{2}\\
 & ={N \choose k}\frac{\left(2t\right)^{k}\, k\,\left(t+1\right)\left(2t+1\right)}{6},\end{align*}
 so that \[
E\left[\xi^{2}\right]=\frac{\sum_{k=1}^{K}{N \choose k}\frac{\left(2t\right)^{k}\, k\,\left(t+1\right)\left(2t+1\right)}{6}}{\sum_{k=1}^{K}{N \choose k}\left(2t\right)^{k}}.\]

Next, substitute $\mathrm{E}\left[\xi^{2}\right]$ into Markov's inequality,
with \[
\mu=\sqrt{\mathrm{E}\left[\xi^{2}\right]},\]
 so that for any $\lambda>1$, it holds that \[
\Pr\left(\xi\ge\lambda\mu\right)\le\frac{1}{\lambda}.\]
 This result implies that at least a $\left(1-1/\lambda\right)$-fraction
of all possible $\mathbf{Cb}$ vectors lie within an $m$-dimensional
ball of radius $\lambda\mu$ around the origin. As a result, one obtains
a sphere packing bound of the form\[
\left(1-\frac{1}{\lambda}\right)\left|\mathcal{B}_{K}\right|\le\left(\frac{\lambda\mu+\frac{d}{2}}{\frac{d}{2}}\right)^{m}.\]
 Note that\[
\mu^{2}=E\left[\xi^{2}\right]\le\frac{K}{3}\left(t+1\right)^{2},\]
 and that \[
\left|\mathcal{B}_{K}\right|\ge\left(\frac{N-K}{K}\right)^{K}\left(2t\right)^{K}.\]
 Consequently, one has \[
\left(1-\frac{1}{\lambda}\right)\left(\frac{N-K}{K}\right)^{K}\left(2t\right)^{K}\le\left(1+\frac{\lambda\sqrt{k}\left(t+1\right)}{d}\right)^{m},\]
 or, equivalently, \begin{align*}
\frac{\log N}{m} & \le\frac{\log K}{2K}+\frac{1}{K}\log\left(\frac{\lambda\left(t+1\right)}{d}+\frac{1}{\sqrt{K}}\right)\\
 & \;-\frac{\log\left(1-\frac{1}{\lambda}\right)}{mK}-\frac{1}{m}\log\left(\frac{1}{K}-\frac{1}{N}\right)-\frac{\log\left(2t\right)}{m}.\end{align*}
 Without loss of generality, we choose $\lambda=2$. The asymptotic
code exponent is therefore upper bounded by \[
\frac{\log K}{2K}\left(1+o_{t,d}\left(1\right)\right),\]
 where \begin{align}
o_{t,d}\left(1\right) & =\frac{2}{\log K}\log\left(\frac{2\left(t+1\right)}{d}+\frac{1}{\sqrt{K}}\right)\overset{K\rightarrow\infty}{\longrightarrow}0\label{eq:o_ub_Euclidean}\end{align}
 if $t=o\left(K\right)$. This proves the upper bound of Theorem \ref{thm:WESC-code-exponent}.

\section{\label{sec:lb-Euclidean}Proof of the Lower Bound for WESCs}

Similarly as for the case of compressive sensing matrix design, we
show that standard Gaussian random matrices, with appropriate scaling,
can be used as codebooks of WESCs. Let $\mathbf{H}\in\mathbb{R}^{m\times N}$
be a standard Gaussian random matrix, and let $\mathbf{h}_{j}$ denote
the $j^{\mathrm{th}}$ column of $\mathbf{H}$. Let $\mathbf{v}_{j}=\mathbf{h}_{j}/\left\Vert \mathbf{h}_{j}\right\Vert _{2}$
and $\mathbf{C}=\left[\mathbf{v}_{1}\cdots\mathbf{v}_{N}\right]$.
Then $\mathbf{C}$ is a codebook with unit $l_{2}$-norm codewords.
Now choose a $\delta>0$ such that $d\left(1+\delta\right)<1$. Let
\begin{equation}
E_{1}=\bigcup_{j=1}^{N}\left\{ \mathbf{H}:\;\frac{1}{\sqrt{m}}\left\Vert \mathbf{h}_{j}\right\Vert _{2}\in\left(1-\delta,1+\delta\right)\right\} \label{eq:Event1-l2}\end{equation}
 be the event that the normalized $l_{2}$-norms of all the columns
of $\mathbf{H}$ concentrate around one. Let \begin{equation}
E_{2}=\bigcup_{\mathcal{B}_{K}\ni\mathbf{b}_{1}\ne\mathbf{b}_{2}\in\mathcal{B}_{K}}\;\left\{ \mathbf{H}:\;\left\Vert \mathbf{C}\left(\mathbf{b}_{1}-\mathbf{b}_{2}\right)\right\Vert _{2}\ge d\right\} .\label{eq:Event2-l2}\end{equation}
 In other words, $E_{2}$ denotes the event that any two different
superpositions of codewords lie at Euclidean distance at least $d$
from each other. In the following, we show that for any \[
R<\frac{\log K}{4K}\left(1+o\left(1\right)\right),\]
 for which $o\left(1\right)$ is given by Equation~(\ref{eq:o_lb_Euclidean_main}),
if \begin{equation}
\underset{\left(m,N\right)\rightarrow\infty}{\lim}\frac{\log N}{m}\le R,\label{eq:WESC-lb-asymptotic-regime}\end{equation}
 then\begin{equation}
\underset{\left(m,N\right)\rightarrow\infty}{\lim}\Pr\left(E_{2}\right)=1.\label{eq:E2-convergence-Euclidean}\end{equation}
 This will establish the lower bound of Theorem~\ref{thm:WESC-code-exponent}.

Note that \[
\Pr\left(E_{2}\right)\ge\Pr\left(E_{2}\bigcap E_{1}\right)=\Pr\left(E_{1}\right)-\Pr\left(E_{1}\bigcap E_{2}^{c}\right).\]
 According to Theorem \ref{thm:L2-column-norm-1}, stated and proved
in the next subsection, one has \[
\underset{\left(m,N\right)\rightarrow\infty}{\lim}\Pr\left(E_{1}\right)=1.\]
 Thus, the desired relation (\ref{eq:E2-convergence-Euclidean}) holds
if \[
\underset{\left(m,N\right)\rightarrow\infty}{\lim}\Pr\left(E_{1}\bigcap E_{2}^{c}\right)=0.\]
 Observe that \[
\mathbf{C}\left(\mathbf{b}_{1}-\mathbf{b}_{2}\right)=\frac{1}{\sqrt{m}}\mathbf{H}\mathbf{b}^{\prime},\]
 where \begin{equation}
\mathbf{b}^{\prime}:=\mathbf{\Lambda}_{\mathbf{H}}\left(\mathbf{b}_{1}-\mathbf{b}_{2}\right),\label{eq:b-prime-L2}\end{equation}
 and \begin{equation}
\mathbf{\Lambda_{H}}=\left[\begin{array}{ccc}
\sqrt{m}/\left\Vert \mathbf{h}_{1}\right\Vert _{2}\\
 & \ddots\\
 &  & \sqrt{m}/\left\Vert \mathbf{h}_{N}\right\Vert _{2}\end{array}\right].\label{eq:Lambda-H-L2}\end{equation}
 By Theorem \ref{eq:o_lb_Euclidean_main} in Section \ref{sub:superposition-Euclidean},
in the asymptotic domain of~(\ref{eq:WESC-lb-asymptotic-regime}),
\begin{align*}
 & \Pr\left(E_{1}\bigcap\left\{ \mathbf{H}:\;\frac{1}{\sqrt{m}}\mathbf{H}\left(\left(1+\delta\right)\mathbf{b}^{\prime}\right)\le d\left(1+\delta\right)\right\} \right)\\
 & =\Pr\left(E_{1}\bigcap\left\{ \mathbf{H}:\;\frac{1}{\sqrt{m}}\mathbf{H}\mathbf{b}^{\prime}\le d\right\} \right)\\
 & =\Pr\left(E_{1}\bigcap\left\{ \mathbf{H}:\;\mathbf{C}\left(\mathbf{b}_{1}-\mathbf{b}_{2}\right)\le d\right\} \right)\\
 & \rightarrow0.\end{align*}
 This establishes the lower bound of Theorem~\ref{thm:WESC-code-exponent}.

\subsection{\label{sub:norm-concentration-Euclidean}Column Norms of $\mathbf{H}$}

In this subsection, we quantify the rate regime in which the Euclidean
norms of all columns of $\mathbf{H}$, when properly normalized, are
concentrated around the value one with high probability.

\begin{thm}
\label{thm:L2-column-norm-1}Let $\mathbf{H}\in\mathbb{R}^{m\times N}$
be a standard Gaussian random matrix, and $\mathbf{h}_{j}$ be the
$j^{\mathrm{th}}$ column of $\mathbf{H}$. 
\begin{enumerate}
\item For a given $\delta\in\left(0,1\right)$, \[
\Pr\left(\left|\frac{1}{m}\left\Vert \mathbf{h}_{j}\right\Vert _{2}^{2}-1\right|>\delta\right)\le2\exp\left(-\frac{m}{4}\delta^{2}\right)\]
 for all $1\le j\le N$. 
\item If $m,N\rightarrow\infty$ simultaneously, so that \[
\underset{\left(m,N\right)\rightarrow\infty}{\lim}\frac{1}{m}\log N<\frac{\delta^{2}}{4},\]
 then it holds that \[
\underset{\left(m,N\right)\rightarrow\infty}{\lim}\Pr\left(\bigcup_{j=1}^{N}\left\{ \left|\frac{1}{m}\left\Vert \mathbf{h}_{j}\right\Vert _{2}^{2}-1\right|>\delta\right\} \right)=0.\]

\end{enumerate}
\end{thm}
\begin{proof}
The first part of this theorem is proved by invoking large deviations
techniques. Note that $\left\Vert \mathbf{h}_{j}\right\Vert _{2}^{2}=\sum_{i=1}^{m}\left|H_{i,j}\right|^{2}$
is chi-square distributed. We have\begin{align}
 & \Pr\left\{ \frac{1}{m}\sum_{i=1}^{m}\left|H_{i,j}\right|^{2}>1+\delta\right\} \nonumber \\
 & \overset{\left(a\right)}{\le}\exp\left\{ -m\left(\alpha\left(1+\delta\right)-\log\mathrm{E}\left[e^{\alpha\left|H_{i,j}\right|^{2}}\right]\right)\right\} \nonumber \\
 & =\exp\left\{ -m\left(\alpha\left(1+\delta\right)+\frac{1}{2}\log\left(1-2\alpha\right)\right)\right\} \nonumber \\
 & \overset{\left(b\right)}{=}\exp\left\{ -\frac{m}{2}\left(\delta-\log\left(1+\delta\right)\right)\right\} ,\label{eq:Gaussian-vector-LD_+delta}\end{align}
 and \begin{align}
 & \Pr\left\{ \frac{1}{m}\sum_{i=1}^{m}\left|H_{i,j}\right|^{2}<1-\delta\right\} \nonumber \\
 & \overset{\left(c\right)}{\le}\exp\left\{ m\left(\alpha\left(1-\delta\right)+\log\mathrm{E}\left[e^{-\alpha\left|H_{i,j}\right|^{2}}\right]\right)\right\} \nonumber \\
 & =\exp\left\{ m\left(\alpha\left(1-\delta\right)-\frac{1}{2}\log\left(1+2\alpha\right)\right)\right\} \nonumber \\
 & \overset{\left(d\right)}{=}\exp\left\{ -\frac{m}{2}\left(-\log\left(1-\delta\right)-\delta\right)\right\} ,\label{eq:Gaussian-vector-LD-delta}\end{align}
 where $\left(a\right)$ and $\left(c\right)$ hold for arbitrary
$\alpha>0$, and $\left(b\right)$ and $\left(d\right)$ are obtained
by specializing $\alpha$ to $\frac{1}{2}\frac{\delta}{1+\delta}$
and $\frac{1}{2}\frac{\delta}{1-\delta}$, respectively. By observing
that \[
-\log\left(1-\delta\right)-\delta>\delta-\log\left(1+\delta\right)>0,\]
 we arrive at \begin{align*}
 & \Pr\left\{ \left|\frac{1}{m}\left\Vert \mathbf{h}_{j}\right\Vert _{2}^{2}-1\right|>\delta\right\} \\
 & \le2\exp\left\{ -\frac{m}{2}\left(\delta-\log\left(1+\delta\right)\right)\right\} \\
 & \le2\exp\left(-\frac{m}{4}\delta^{2}\right).\end{align*}

The second part of the claimed result is proved by applying the union
bound, i.e. \begin{align*}
 & \Pr\left\{ \bigcup_{j=1}^{N}\left\{ \left|\frac{1}{m}\left\Vert \mathbf{h}_{j}\right\Vert ^{2}-1\right|>\delta\right\} \right\} \\
 & \le2N\exp\left\{ -\frac{m}{4}\delta^{2}\right\} \\
 & =\exp\left\{ -m\left(\frac{\delta^{2}}{4}-\left(\frac{1}{m}\log N+\frac{\log2}{m}\right)\right)\right\} .\end{align*}
 This completes the proof of Theorem 4. 
\end{proof}
\vspace{0.05in}

\subsection{\label{sub:superposition-Euclidean}The Distance Between Two Different
Superpositions}

This section is devoted to identifying the rate regime in which any
pair of different superpositions is at sufficiently large Euclidean
distance. The main result is presented in Theorem \ref{thm:L2-superposition-2}
at the end of this subsection. Since the proof of this theorem is
rather technical, and since it involves complicated notation, we first
prove a simplified version of the result, stated in Theorem~\ref{thm:L2-superposition-1}.

\begin{thm}
\label{thm:L2-superposition-1}Let $\mathbf{H}\in\mathbb{R}^{m\times N}$
be a standard Gaussian random matrix, and let $\delta\in\left(0,1\right)$
be fixed. For sufficiently large $K$, if \[
\underset{\left(m,N\right)\rightarrow\infty}{\lim}\frac{\log N}{m}<\frac{\log K}{4K}\left(1+o_{\delta}\left(1\right)\right)\]
 where the exact expression for $o_{\delta}\left(1\right)$ given
by Equation~(\ref{eq:o_lb_Euclidean_1}), then \begin{equation}
\underset{\left(m,N\right)\rightarrow\infty}{\lim}\Pr\left(\frac{1}{m}\left\Vert \mathbf{H}\mathbf{b}\right\Vert _{2}^{2}\le\delta^{2}\right)=0,\label{eq:L2-delta}\end{equation}
 for all $\mathbf{b}\in B_{2t}^{N}$ such that $\left\Vert \mathbf{b}\right\Vert _{0}\le2K$. 
\end{thm}
\begin{proof}
By the union bound, we have \begin{align}
 & \Pr\left(\bigcup_{\left\Vert \mathbf{b}\right\Vert _{0}\le2K}\;\left\{ \mathbf{H}:\;\frac{1}{m}\left\Vert \mathbf{H}\mathbf{b}\right\Vert _{2}^{2}\le\delta^{2}\right\} \right)\nonumber \\
 & =\Pr\left(\bigcup_{k=1}^{2K}\;\bigcup_{\left\Vert \mathbf{b}\right\Vert _{0}=k}\;\left\{ \mathbf{H}:\;\frac{1}{m}\left\Vert \mathbf{H}\mathbf{b}\right\Vert _{2}^{2}\le\delta^{2}\right\} \right)\nonumber \\
 & \le\sum_{k=1}^{2K}{N \choose k}\left(4t\right)^{k}\Pr\left(\frac{1}{m}\left\Vert \mathbf{H}\mathbf{b}\right\Vert _{2}^{2}\le\delta^{2},\;\left\Vert \mathbf{b}\right\Vert _{0}=k\right).\label{eq:L2-union-bd}\end{align}

We shall upper bound the probability \[
\Pr\left(\frac{1}{m}\left\Vert \mathbf{H}\mathbf{b}\right\Vert _{2}^{2}\le\delta^{2},\;\left\Vert \mathbf{b}\right\Vert _{0}=k\right)\]
 for each $k=1,\cdots,2K$. From Chernoff's inequality, for all $\alpha>0$,
it holds that \begin{align*}
 & \Pr\left(\frac{1}{m}\left\Vert \mathbf{H}\mathbf{b}\right\Vert _{2}^{2}\le\delta^{2},\;\left\Vert \mathbf{b}\right\Vert _{0}=k\right)\\
 & \le\exp\left\{ m\left(\alpha\delta^{2}+\log\mathrm{E}\left[e^{-\alpha\left(\mathbf{H}_{i,\cdot}\mathbf{b}\right)^{2}}\right]\right)\right\} ,\end{align*}
 where $\mathbf{H}_{i,\cdot}$ is the $i^{\mathrm{th}}$ row of the
$\mathbf{H}$ matrix. Furthermore, \begin{align*}
 & \mathrm{E}\left[e^{-\alpha\left(\mathbf{H}_{i,\cdot}\mathbf{b}\right)^{2}}\right]\\
 & =\mathrm{E}\left[\exp\left\{ -\alpha\left\Vert \mathbf{b}\right\Vert _{2}^{2}\left(\mathbf{H}_{i,\cdot}\left(\mathbf{b}/\left\Vert \mathbf{b}\right\Vert _{2}\right)\right)^{2}\right\} \right]\\
 & \overset{\left(a\right)}{\le}\mathrm{E}\left[\exp\left\{ -\alpha k\left(\mathbf{H}_{i,\cdot}\left(\mathbf{b}/\left\Vert \mathbf{b}\right\Vert _{2}\right)\right)^{2}\right\} \right]\\
 & \overset{\left(b\right)}{\le}-\frac{1}{2}\log\left(1+2\alpha k\right),\end{align*}
 where $\left(a\right)$ follows from the fact that $\left\Vert \mathbf{b}\right\Vert _{2}^{2}\ge k$
for all $\mathbf{b}\in B_{2t}^{N}$ such that $\left\Vert \mathbf{b}\right\Vert _{0}=k$,
and where $\left(b\right)$ holds because $\mathbf{H}_{i,\cdot}\left(\mathbf{b}/\left\Vert \mathbf{b}\right\Vert _{2}\right)$
is a standard Gaussian random variable. Let \[
\alpha=\frac{1}{2k}\frac{k-\delta^{2}}{\delta^{2}}=\frac{1}{2\delta^{2}}-\frac{1}{2k}.\]
 Then \begin{align*}
 & \Pr\left(\frac{1}{m}\left\Vert \mathbf{H}\mathbf{b}\right\Vert _{2}^{2}\le\delta^{2},\;\left\Vert \mathbf{b}\right\Vert _{0}=k\right)\\
 & \le\exp\left\{ m\left(\left(\frac{1}{2}-\frac{\delta^{2}}{2k}\right)-\frac{1}{2}\log\frac{k}{\delta^{2}}\right)\right\} \\
 & =\exp\left\{ -\frac{m}{2}\left(\log k-\log\delta^{2}+\frac{\delta^{2}}{k}-1\right)\right\} .\end{align*}

Substituting the above expression into the union bound gives \begin{align*}
 & \Pr\left(\bigcup_{\left\Vert \mathbf{b}\right\Vert _{0}\le2K}\;\left\{ \mathbf{H}:\;\frac{1}{m}\left\Vert \mathbf{H}\mathbf{b}\right\Vert _{2}^{2}\le\delta^{2}\right\} \right)\\
 & \le\sum_{k=1}^{2K}\exp\left\{ -\frac{m}{2}\left(\log k-\log\delta^{2}+\frac{\delta^{2}}{k}-1\right.\right.\\
 & \qquad\qquad\qquad\left.\left.-\frac{2k}{m}\log N-\frac{2k}{m}\log\left(4t\right)\right)\right\} \\
 & \le\sum_{k=1}^{2K}\exp\left\{ -mk\left(\frac{\log k}{2k}+\frac{\delta^{2}/k-\log\delta^{2}-1}{2k}\right.\right.\\
 & \qquad\qquad\qquad\left.\left.-\frac{1}{m}\log N-\frac{1}{m}\log\left(4t\right)\right)\right\} .\end{align*}
 Now, let $K$ be sufficiently large so that \begin{align*}
 & \frac{\log2K}{4K}+\frac{\delta^{2}/\left(2K\right)-\log\delta^{2}-1}{4K}\\
 & =\underset{1\le k\le2K}{\min}\;\left(\frac{\log k}{2k}+\frac{\delta^{2}/k-\log\delta^{2}-1}{2k}\right).\end{align*}
 If \[
\underset{\left(m,N\right)\rightarrow\infty}{\lim}\frac{\log N}{m}<\frac{\log K}{4K}\left(1+o_{\delta}\left(1\right)\right),\]
 where \begin{equation}
o_{\delta}\left(1\right)=\frac{\log2+\delta^{2}/\left(2K\right)-\log\delta^{2}-1}{\log K},\label{eq:o_lb_Euclidean_1}\end{equation}
 then\[
\underset{\left(m,N\right)\rightarrow\infty}{\lim}\Pr\left(\frac{1}{m}\left\Vert \mathbf{H}\mathbf{b}\right\Vert _{2}^{2}\le\delta^{2}\right)=0\]
 for all $\mathbf{b}\in B_{2t}^{N}$ such that $\left\Vert \mathbf{b}\right\Vert _{0}\le2K$.
This completes the proof of the claimed result. 
\end{proof}
\vspace{0.05in}

Based on Theorem \ref{thm:L2-superposition-1}, the asymptotic region
in which $\Pr\left(E_{2}^{c}\bigcap E_{1}\right)\rightarrow0$ is
characterized in below.

\vspace{0.05in}

\begin{thm}
\label{thm:L2-superposition-2}Let $\mathbf{H}\in\mathbb{R}^{m\times N}$
be a standard Gaussian random matrix, and for a given $\mathbf{H}$,
let $\mathbf{\Lambda_{H}}$ be as defined in (\ref{eq:Lambda-H-L2}).
For a given $d\in\left(0,1\right)$, choose a $\delta>0$ such that
$d\left(1+\delta\right)<1$. Define the set $E_{1}$ as in (\ref{eq:Event1-l2}).
For sufficiently large $K$, if \[
\underset{\left(m,N\right)\rightarrow\infty}{\lim}\frac{\log N}{m}<\frac{\log K}{4K}\left(1+o\left(1\right)\right)\]
 where \begin{equation}
o\left(1\right)=\frac{\log2-1}{\log K},\label{eq:o_lb_Euclidean_main}\end{equation}
 then \begin{equation}
\underset{\left(m,N\right)\rightarrow\infty}{\lim}\Pr\left(\frac{1}{m}\left\Vert \mathbf{H\Lambda_{H}}\left(\mathbf{b}_{1}-\mathbf{b}_{2}\right)\right\Vert _{2}^{2}\le d^{2},\; E_{1}\right)=0,\label{eq:L2-delta}\end{equation}
 for all pairs of $\mathbf{b}_{1},\mathbf{b}_{2}\in\mathcal{B}_{K}$
such that $\mathbf{b}_{1}\ne\mathbf{b}_{2}$. 
\end{thm}
\begin{proof}
The proof is analogous to that of Theorem \ref{thm:L2-column-norm-1}
with minor changes. Let $\mathbf{b}^{\prime}=\mathbf{\Lambda_{H}}\left(\mathbf{b}_{1}-\mathbf{b}_{2}\right)$.
On the set $E_{1}$, since \[
\frac{1}{\sqrt{m}}\left\Vert \mathbf{h}_{j}\right\Vert _{2}\le1+\delta\]
 for all $1\le j\le N$, the nonzero entries of $\left(1+\delta\right)\mathbf{b}^{\prime}$
satisfy \[
\left|\left(1+\delta\right)b_{i}^{\prime}\right|\ge1.\]
 Replace $\mathbf{b}$ in Theorem \ref{thm:L2-superposition-1} with
$\left(1+\delta\right)\mathbf{b}^{\prime}$. All the arguments in
the proof of Theorem \ref{thm:L2-superposition-1} are still valid,
except that the higher order term is changed to \begin{align*}
 & o_{d\left(1+\delta\right)}\left(1\right)\\
 & =\frac{\log2+d^{2}\left(1+\delta\right)^{2}/\left(2K\right)-1-\log\left(d^{2}\left(1+\delta\right)^{2}\right)}{\log K}\\
 & \ge\frac{\log2-1}{\log K}.\end{align*}
 This completes the proof of the theorem. 
\end{proof}

\section{\label{sec:lb-L1WSC}Proof of the Lower Bound for $l_{1}$-WSCs}

The proof is similar to that of the lower bound for WESCs. Let $\mathbf{A}\in\mathbb{R}^{m\times N}$
be a standard Gaussian random matrix, and let $\mathbf{H}$ be the
matrix with entries \[
H_{i,j}=\frac{\sqrt{2\pi}}{2}A_{i,j},\;1\le i\le m,\;1\le j\le N.\]
 Once more, let $\mathbf{h}_{j}$ be the $j^{\mathrm{th}}$ column
of $\mathbf{H}$. Let $\mathbf{v}_{j}=\mathbf{h}_{j}/\left\Vert \mathbf{h}_{j}\right\Vert _{1}$
and $\mathbf{C}=\left[\mathbf{v}_{1}\cdots\mathbf{v}_{N}\right]$.
Then $\mathbf{C}$ is a codebook with unit $l_{1}$-norm codewords.
Now choose a $\delta>0$ such that $d\left(1+\delta\right)<1$. Let
\begin{equation}
E_{1}=\bigcup_{j=1}^{N}\left\{ \mathbf{H}:\;\frac{1}{m}\left\Vert \mathbf{h}_{j}\right\Vert _{1}\in\left(1-\delta,1+\delta\right)\right\} ,\label{eq:Event1-l1}\end{equation}
 and \begin{equation}
E_{2}=\bigcup_{\mathcal{B}_{K}\ni\mathbf{b}_{1}\ne\mathbf{b}_{2}\in\mathcal{B}_{K}}\;\left\{ \mathbf{H}:\;\left\Vert \mathbf{C}\left(\mathbf{b}_{1}-\mathbf{b}_{2}\right)\right\Vert _{1}\ge d\right\} .\label{eq:Event2-l1}\end{equation}
 We consider the asymptotic regime where\[
\underset{\left(m,N\right)\rightarrow\infty}{\lim}\frac{\log N}{m}\le R,\]
 \[
R<\frac{\log K}{4K}\left(1+o\left(1\right)\right),\]
 and $o_{d}\left(1\right)$ is given in Equation~(\ref{eq:o_lb_L1_WSC-2}).
Theorem \ref{thm:L1-codeword-norm} in Section \ref{sub:L1-Codeword-Norms}
suggests that \[
\underset{\left(m,N\right)\rightarrow\infty}{\lim}\Pr\left(E_{1}\right)=1,\]
 while Theorem \ref{thm:L1-superposition-1} in Section \ref{sub:L1-superposition-1}
shows that \[
\underset{\left(m,N\right)\rightarrow\infty}{\lim}\Pr\left(E_{1}\bigcap E_{2}^{c}\right)=0.\]
 Therefore,\[
\underset{\left(m,N\right)\rightarrow\infty}{\lim}\Pr\left(E_{2}\right)\ge\underset{\left(m,N\right)\rightarrow\infty}{\lim}\Pr\left(E_{1}\right)-\Pr\left(E_{1}\bigcap E_{2}^{c}\right)=1.\]
 This result implies the lower bound of Theorem~\ref{thm:L1WSC-code-exponent}.

\subsection{\label{sub:L1-Codeword-Norms}Column Norms of $\mathbf{H}$}

The following theorem quantifies the rate regime in which the $l_{1}$-norms
of all columns of $\mathbf{H}$, with proper normalization, are concentrated
around one with high probability.

\begin{thm}
\label{thm:L1-codeword-norm}Let $\mathbf{A}\in\mathbb{R}^{m\times N}$
be a standard Gaussian random matrix. Let $\mathbf{H}$ be the matrix
with entries \[
H_{i,j}=\frac{\sqrt{2\pi}}{2}A_{i,j},\;1\le i\le m,\;1\le j\le N.\]
 Let $\mathbf{h}_{j}$ be the $j^{\mathrm{th}}$ column of $\mathbf{H}$. 
\begin{enumerate}
\item For a given $\delta\in\left(0,1\right)$, \[
\Pr\left(\left|\frac{1}{m}\left\Vert \mathbf{h}_{j}\right\Vert _{1}-1\right|>\delta\right)\le c_{1}e^{-mc_{2}\delta^{2}}\]
 for some positive constant $c_{1}$ and $c_{2}$; 
\item Let $m,N\rightarrow\infty$ simultaneously, with \[
\underset{\left(m,N\right)\rightarrow\infty}{\lim}\frac{1}{m}\log N<c_{2}\delta^{2}.\]
 The it holds that \[
\underset{\left(m,N\right)\rightarrow\infty}{\lim}\Pr\left(\bigcup_{j=1}^{N}\left\{ \left|\frac{1}{m}\left\Vert \mathbf{h}_{j}\right\Vert _{1}-1\right|>\delta\right\} \right)=0.\]

\end{enumerate}
\end{thm}
\begin{proof}
$ $ 
\begin{enumerate}
\item Since $A_{i,j}$ is a standard Gaussian random variable, $\left|A_{i,j}\right|$
is a Subgaussian distributed random variable, and $\mathrm{E}\left[\left|A_{i,j}\right|\right]=\frac{2}{\sqrt{2\pi}}$.
According to Proposition \ref{pro:subgaussian-shift} in Appendix
\ref{sub:SubGaussian-stuff}, $\left|A_{i,j}\right|-\frac{2}{\sqrt{2\pi}}$
is a Subgaussian random variable with zero mean. A direct application
of Theorem~\ref{thm:subgaussian-large-deviation} stated in Appendix~\ref{sub:SubGaussian-stuff}
gives \begin{align*}
 & \Pr\left(\left|\frac{1}{m}\left\Vert \mathbf{h}_{j}\right\Vert _{1}-1\right|>\delta\right)\\
 & =\Pr\left(\left|\sum_{i=1}^{m}\left(\left|A_{i,j}\right|-\frac{2}{\sqrt{2\pi}}\right)\right|>\frac{2m\delta}{\sqrt{2\pi}}\right)\\
 & \le c_{1}\exp\left(-c_{2}m\delta^{2}\right),\end{align*}
 which proves claim 1). 
\item This part is proved by using the union bound: first, note that \begin{align*}
 & \Pr\left(\bigcup_{j=1}^{N}\left\{ \left|\frac{1}{m}\left\Vert \mathbf{h}_{j}\right\Vert _{1}-1\right|>\delta\right\} \right)\\
 & \le\exp\left(-mc_{2}\delta^{2}+\log c_{1}+\log N\right)\\
 & =\exp\left\{ -m\left(c_{2}\delta^{2}-\frac{1}{m}\log c_{1}-\frac{1}{m}\log N\right)\right\} .\end{align*}
 If \[
\underset{\left(m,N\right)\rightarrow\infty}{\lim}\frac{1}{m}\log N<c_{2}\delta^{2},\]
 then one has \[
\underset{\left(m,N\right)\rightarrow\infty}{\lim}\Pr\left(\bigcup_{j=1}^{N}\left\{ \left|\frac{1}{m}\left\Vert \mathbf{h}_{j}\right\Vert _{1}-1\right|>\delta\right\} \right)=0.\]
 This completes the proof of claim 2). 
\end{enumerate}
\end{proof}

\subsection{\label{sub:L1-superposition-1}The Distance Between Two Different
Superpositions}

Similarly to the analysis performed for WESCs, we start with a proof
of a simplified version of the result needed in order to simplify
tedious notation. We then explain how to establish the proof of Theorem
\ref{thm:L1-superposition-2} by modifying some of the steps of the
simplified theorem.

\begin{thm}
\label{thm:L1-superposition-1}Let $\mathbf{A}\in\mathbb{R}^{m\times N}$
be a standard Gaussian random matrix. Let $\mathbf{H}$ be the matrix
with entries \[
H_{i,j}=\frac{\sqrt{2\pi}}{2}A_{i,j},\;1\le i\le m,\;1\le j\le N.\]
 Let $\delta\in\left(0,1\right)$ be given. For sufficiently large
$K$, if \[
\underset{\left(m,N\right)\rightarrow\infty}{\lim}\frac{\log N}{m}<\frac{\log K}{4K}\left(1+o_{\delta}\left(1\right)\right),\]
 where \begin{equation}
o_{\delta}\left(1\right)=\frac{2}{\log K}\left(\log\frac{\pi}{2\delta}-1\right),\label{eq:o_lb_L1-1}\end{equation}
 then \begin{equation}
\underset{\left(m,N\right)\rightarrow\infty}{\lim}\Pr\left(\frac{1}{m}\left\Vert \mathbf{H}\mathbf{b}\right\Vert _{1}\le\delta\right)=0\label{eq:L1-delta}\end{equation}
 for all $\mathbf{b}\in B_{2t}^{N}$ such that $\left\Vert \mathbf{b}\right\Vert _{0}\le2K$. 
\end{thm}
\begin{proof}
The proof starts by using the union bound, as \begin{align}
 & \Pr\left(\bigcup_{\left\Vert \mathbf{b}\right\Vert _{0}\le2K}\;\left\{ \mathbf{H}:\;\frac{1}{m}\left\Vert \mathbf{H}\mathbf{b}\right\Vert _{1}\le\delta\right\} \right)\nonumber \\
 & \le\sum_{k=1}^{2K}{N \choose k}\left(4t\right)^{k}\Pr\left(\frac{1}{m}\left\Vert \mathbf{H}\mathbf{b}\right\Vert _{1}\le\delta,\;\left\Vert \mathbf{b}\right\Vert _{0}=k\right).\label{eq:L1-union-bd}\end{align}
 To estimate the above upper bound, we have to upper bound the probability
\[
\Pr\left(\frac{1}{m}\left\Vert \mathbf{H}\mathbf{b}\right\Vert _{1}\le\delta,\;\left\Vert \mathbf{b}\right\Vert _{0}=k\right),\]
 for each $k=1,\cdots,2K$. Let us derive next an expression for such
an upper bound that holds for arbitrary values of $k\ge1$.

Note that \begin{align}
 & \mathrm{E}\left[e^{-\alpha\left|\sum_{j=1}^{k}b_{j}A_{i,j}\right|}\right]=\int_{0}^{\infty}\frac{2}{\sqrt{2\pi}\left\Vert \mathbf{b}\right\Vert _{2}}e^{-\frac{x^{2}}{2\left\Vert \mathbf{b}\right\Vert _{2}^{2}}}e^{-\alpha x}\cdot dx\nonumber \\
 & \overset{\left(a\right)}{=}\int_{0}^{\infty}\frac{2}{\sqrt{2\pi}}e^{-\frac{x^{2}}{2}}e^{-\alpha\left\Vert \mathbf{b}\right\Vert _{2}x}\cdot dx\nonumber \\
 & =e^{\frac{\alpha^{2}\left\Vert \mathbf{b}\right\Vert _{2}^{2}}{2}}\int_{0}^{\infty}\frac{2}{\sqrt{2\pi}}\exp\left(-\frac{\left(x+\alpha\left\Vert \mathbf{b}\right\Vert _{2}\right)^{2}}{2}\right)\cdot dx\nonumber \\
 & \overset{\left(b\right)}{=}e^{\frac{\alpha^{2}\left\Vert \mathbf{b}\right\Vert _{2}^{2}}{2}}\int_{\alpha\left\Vert \mathbf{b}\right\Vert _{2}}^{\infty}\frac{2}{\sqrt{2\pi}}\exp\left(-\frac{x^{2}}{2}\right)\cdot dx\nonumber \\
 & \le e^{\frac{\alpha^{2}\left\Vert \mathbf{b}\right\Vert _{2}^{2}}{2}}\int_{\alpha\left\Vert \mathbf{b}\right\Vert _{2}}^{\infty}\frac{x}{\alpha\left\Vert \mathbf{b}\right\Vert _{2}}\cdot\frac{2}{\sqrt{2\pi}}\exp\left(-\frac{x^{2}}{2}\right)\cdot dx\nonumber \\
 & =\frac{1}{\alpha\left\Vert \mathbf{b}\right\Vert _{2}}\cdot\frac{2}{\sqrt{2\pi}}e^{\frac{\alpha^{2}\left\Vert \mathbf{b}\right\Vert _{2}^{2}}{2}}e^{-\frac{\alpha^{2}\left\Vert \mathbf{b}\right\Vert _{2}^{2}}{2}}\nonumber \\
 & \overset{\left(c\right)}{\le}\frac{1}{\alpha}\frac{2}{\sqrt{2\pi k}},\label{eq:MGF_Gaussian_case}\end{align}
 where $\left(a\right)$ and $\left(b\right)$ follow from the change
of variables $x^{\prime}=x/\left\Vert \mathbf{b}\right\Vert _{2}$
and $x^{\prime}=x+\alpha\left\Vert \mathbf{b}\right\Vert _{2}$, respectively.
Inequality $\left(c\right)$ holds based on the assumption that $\left\Vert \mathbf{b}\right\Vert _{2}\ge k$.
As a result, \begin{align}
 & \Pr\left(\frac{1}{m}\sum_{i=1}^{m}\left|\sum_{j}b_{j}H_{i,j}\right|\le\delta\right)\nonumber \\
 & =\Pr\left(\frac{1}{m}\sum_{i=1}^{m}\left|\sum_{j=1}^{k}b_{j}A_{i,j}\right|\le\frac{2}{\sqrt{2\pi}}\delta\right)\nonumber \\
 & \le\exp\left\{ m\left(\alpha\frac{2\delta}{\sqrt{2\pi}}+\log\mathrm{E}\left[e^{-\alpha\left|\sum_{j}b_{j}H_{j}\right|}\right]\right)\right\} \nonumber \\
 & \le\exp\left\{ m\left(\alpha\frac{2\delta}{\sqrt{2\pi}}+\log\left(\frac{2}{\sqrt{2\pi k}}\frac{1}{\alpha}\right)\right)\right\} \nonumber \\
 & =\exp\left\{ m\left(\alpha\frac{2\delta}{\sqrt{2\pi}}-\log\left(\alpha\sqrt{\frac{\pi k}{2}}\right)\right)\right\} \nonumber \\
 & =\exp\left\{ m\left(1-\log\left(\sqrt{k}\frac{\pi}{2\delta}\right)\right)\right\} ,\label{eq:ub-L1-k}\end{align}
 where the last equality is obtained by specializing $\alpha=\sqrt{2\,\pi}/2\,\delta$.

The upper bound in~(\ref{eq:ub-L1-k}) is useful only when it is
less than one, or equivalently, \begin{equation}
\log\left(\sqrt{k}\frac{\pi}{2\delta}\right)>1.\label{eq:useful-condition}\end{equation}
 For any $\delta\in\left(0,1\right)$, if $k\ge4$, inequality~(\ref{eq:useful-condition})
holds. Thus, for any $k\ge4$, \begin{align}
 & {N \choose k}\left(4t\right)^{k}\Pr\left(\frac{1}{m}\left\Vert \mathbf{H}\mathbf{b}\right\Vert _{1}\le\delta,\;\left\Vert \mathbf{b}\right\Vert _{0}=k\right)\nonumber \\
 & \le\exp\left\{ -mk\left(\frac{\log k}{2k}\left(1+o_{\delta}\left(1\right)\right)-\frac{\log\left(4t\right)}{m}-\frac{\log N}{m}\right)\right\} \label{eq:ub-l1-general-k}\\
 & \rightarrow0,\nonumber \end{align}
 as $\left(m,N\right)\rightarrow\infty$ with\[
\underset{\left(m,N\right)\rightarrow\infty}{\lim}\frac{\log N}{m}<\frac{\log k}{2k}\left(1+o_{\delta}\left(1\right)\right),\]
 where\[
o_{\delta}\left(1\right)=\frac{2}{\log k}\left(\log\frac{\pi}{2\delta}-1\right).\]

Another upper bound is needed for $k=1,2,3$. For a fixed $k$ taking
one of these values, \begin{align*}
 & \Pr\left(\frac{1}{m}\left\Vert \mathbf{H}\mathbf{b}\right\Vert _{1}\le\delta\right)\\
 & =\Pr\left(\frac{1}{m}\sum_{i}\left|\sum_{j}A_{i,j}b_{j}\right|<\frac{2}{\sqrt{2\pi}}\delta\right)\\
 & =\Pr\left(\sum_{i}\left(\left|\sum_{j}A_{i,j}b_{j}\right|-\frac{2\left\Vert \mathbf{b}\right\Vert _{2}}{\sqrt{2\pi}}\right)<\frac{2m}{\sqrt{2\pi}}\left(\delta-\left\Vert \mathbf{b}\right\Vert _{2}\right)\right).\end{align*}
 It is straightforward to verify that $\sum_{j}A_{i,j}b_{j}$ is Gaussian
and that \[
\mathrm{E}\left[\left|\sum_{j}A_{i,j}b_{j}\right|\right]=\frac{2\left\Vert \mathbf{b}\right\Vert _{2}}{\sqrt{2\pi}}.\]
 Thus \[
\sum_{i}\left(\left|\sum_{j}A_{i,j}b_{j}\right|-\frac{2\left\Vert \mathbf{b}\right\Vert _{2}}{\sqrt{2\pi}}\right)\]
 is a sum of independent zero-mean subgaussian random variables. Furthermore,
$\left\Vert \mathbf{b}\right\Vert _{2}\in\left[\sqrt{k},\sqrt{2k}t\right]$
and therefore, $\delta-\left\Vert \mathbf{b}\right\Vert _{2}<0$.
Hence, we can apply Theorem \ref{thm:subgaussian-large-deviation}
of Appendix \ref{sub:SubGaussian-stuff}: as a result, there exist
positive constants $c_{3,k}$ and $c_{4,k}$ such that \begin{align*}
 & \Pr\left(\frac{1}{m}\left\Vert \mathbf{H}\mathbf{b}\right\Vert _{1}\le\delta\right)\\
 & \le c_{3,k}\exp\left(-c_{4,k}m\left(\delta-\left\Vert \mathbf{b}\right\Vert _{2}\right)^{2}\right)\\
 & \le c_{3,k}\exp\left(-c_{4,k}m\left(\sqrt{k}-\delta\right)^{2}\right).\end{align*}
 Note that the values of $c_{3,k}$ and $c_{4,k}$ depend on $k$.
Consequently, \begin{align}
 & {N \choose k}\left(4t\right)^{k}\Pr\left(\frac{1}{m}\left\Vert \mathbf{H}\mathbf{b}\right\Vert _{1}\le\delta,\;\left\Vert \mathbf{b}\right\Vert _{0}=k\right)\nonumber \\
 & \le c_{3,k}\exp\left\{ -mk\left(c_{4,k}\left(1-\frac{\delta}{\sqrt{k}}\right)^{2}\right.\right.\nonumber \\
 & \qquad\qquad\left.\left.\phantom{\left(\frac{a}{\sqrt{b}}\right)^{2}}-\frac{\log\left(4t\right)}{m}-\frac{\log N}{m}\right)\right\} \label{eq:ub-k-le-3}\\
 & \rightarrow0\nonumber \end{align}
 as $m,N\rightarrow\infty$ with\[
\underset{\left(m,N\right)\rightarrow\infty}{\lim}\frac{\log N}{m}<c_{4,k}\left(1-\frac{\delta}{\sqrt{k}}\right)^{2}.\]

Finally, substitute the upper bounds of~(\ref{eq:ub-l1-general-k})
and (\ref{eq:ub-k-le-3}) into the union bound of Equation~(\ref{eq:L1-union-bd}).
If $K$ is large enough so that \[
\frac{\log K}{4K}\left(1+o_{\delta}\left(1\right)\right)<c_{4,k}\left(1-\frac{\delta}{\sqrt{k}}\right)^{2}\;\mathrm{for\; all}\; k=1,2,3,\]
 and if \[
\frac{\log K}{4K}\left(1+o_{\delta}\left(1\right)\right)\le\underset{4\le k\le2K}{\min}\;\frac{\log k}{2k}\left(1+o_{\delta}\left(1\right)\right),\]
 where $o_{\delta}\left(1\right)$ is as given in (\ref{eq:o_lb_L1-1}),
then the desired result (\ref{eq:L1-delta}) holds. 
\end{proof}
\vspace{0.05in}

Based on Theorem \ref{thm:L1-superposition-1}, we are ready to characterize
the asymptotic region in which $\Pr\left(E_{2}^{c}\bigcap E_{1}\right)\rightarrow0$.

\begin{thm}
\label{thm:L1-superposition-2}Define $\mathbf{A}$ and $\mathbf{H}$
as in Theorem \ref{thm:L1-superposition-2}. For a given $\mathbf{H},$
define the diagonal matrix \[
\mathbf{\Lambda_{H}}=\left[\begin{array}{ccc}
m/\left\Vert \mathbf{h}_{1}\right\Vert _{1} &  & 0\\
 & \ddots\\
0 &  & m/\left\Vert \mathbf{h}_{N}\right\Vert _{1}\end{array}\right].\]
 For a given $d\in\left(0,1\right)$, choose a $\delta>0$ such that
$d\left(1+\delta\right)<1$. Define the set $E_{1}$ as in (\ref{eq:Event1-l1}).
For sufficiently large $K$, if \[
\underset{\left(m,N\right)\rightarrow\infty}{\lim}\frac{\log N}{m}<\frac{\log K}{4K}\left(1+o\left(1\right)\right),\]
 where \begin{equation}
o\left(1\right)=\frac{2}{\log K}\left(\log\pi-1-\log2\right),\label{eq:o_lb_L1_WSC-2}\end{equation}
 then\[
\underset{\left(m,N\right)\rightarrow\infty}{\lim}\Pr\left(\frac{1}{m}\left\Vert \mathbf{H}\mathbf{\Lambda_{H}}\left(\mathbf{b}_{1}-\mathbf{b}_{2}\right)\right\Vert _{1}\le d,\; E_{1}\right)=0\]
 for all pairs of $\mathbf{b}_{1},\mathbf{b}_{2}\in\mathcal{B}_{K}$
such that $\mathbf{b}_{1}\ne\mathbf{b}_{2}$. 
\end{thm}
\begin{proof}
Let $\mathbf{b}^{\prime}=\mathbf{\Lambda_{H}}\left(\mathbf{b}_{1}-\mathbf{b}_{2}\right)$.
On the set $E_{1}$, since \[
\frac{1}{\sqrt{m}}\left\Vert \mathbf{h}_{j}\right\Vert _{1}\le1+\delta,\]
 for all $1\le j\le N$, all the nonzero entries of $\left(1+\delta\right)\mathbf{b}^{\prime}$
satisfy \[
\left|\left(1+\delta\right)b_{i}^{\prime}\right|\ge1.\]
 Replace $\mathbf{b}$ in Theorem \ref{thm:L1-superposition-1} with
$\left(1+\delta\right)\mathbf{b}^{\prime}$. All arguments used in
the proof of Theorem \ref{thm:L1-superposition-1} are still valid,
except that now, the higher order term~(\ref{eq:o_lb_L1-1}) in the
asymptotic expression reads as \begin{align*}
 & o_{d\left(1+\delta\right)}\left(1\right)\\
 & =\frac{2}{\log K}\left(\log\pi-1-\log2-\log\left(d\left(1+\delta\right)\right)\right)\\
 & \ge\frac{2}{\log K}\left(\log\pi-1-\log2\right).\end{align*}
 This completes the proof. 
\end{proof}

\section{\label{sec:lb-ngL1WSC}Proof of the Lower Bound for Nonnegative $l_{1}$-WSCs}

The proof follows along the same lines as the one described for $l_{1}$-WSCs.
However, there is a serious technical difficulty associated with the
analysis of nonnegative $l_{1}$-WSCs. Let $\mathbf{A}\in\mathbb{R}^{m\times N}$
be a standard Gaussian random matrix. For general $l_{1}$-WSCs, we
let \[
H_{i,j}=\frac{\sqrt{2\pi}}{2}A_{i,j},\;1\le i\le m,\;1\le j\le N,\]
 and therefore, \[
\sum_{j=1}^{N}H_{i,j}b_{j}\]
 is a Gaussian random variable, whose parameters are easy to determine.
However, for nonnegative $l_{1}$-WSCs, one has to set \begin{equation}
H_{i,j}=\frac{\sqrt{2\pi}}{2}\left|A_{i,j}\right|,\;1\le i\le m,\;1\le j\le N.\label{eq:H-def-ngL1WSC}\end{equation}
 Since the random variables $H_{i,j}$s are not Gaussian, but rather
one-sided Gaussian, \[
\sum_{j=1}^{N}H_{i,j}b_{j}\]
 is not Gaussian distributed, and it is complicated to exactly characterize
its properties.

Nevertheless, we can still define $E_{1}$ and $E_{2}$ as in Equations~(\ref{eq:Event1-l1})
and~(\ref{eq:Event2-l1}). The results of Theorem~\ref{thm:L1-codeword-norm}
are still valid under the non-negativity assumption: the norms of
all $\mathbf{H}$ columns concentrate around one in the asymptotic
regime described in Theorem~\ref{thm:L1-codeword-norm}. The key
step in the proof of the lower bound is to identify the asymptotic
region in which any two different superpositions are sufficiently
separated in terms of the $l_{1}$-distance. We therefore use an approach
similar to the one we invoked twice before: we first prove a simplified
version of the claim, and then proceed with proving the needed result
by introducing some auxiliary variables and notation.

\begin{thm}
\label{thm:L1-superposition-ngWSC-1}Let $\mathbf{A}\in\mathbb{R}^{m\times N}$
be a standard Gaussian random matrix. Let $\mathbf{H}$ be the matrix
with entries \[
H_{i,j}=\frac{\sqrt{2\pi}}{2}\left|A_{i,j}\right|,\;1\le i\le m,\;1\le j\le N.\]
 Let $\delta\in\left(0,1\right)$ be given. For a given sufficiently
large $K$, if \[
\underset{\left(m,N\right)\rightarrow\infty}{\lim}\frac{\log N}{m}<\frac{\log K}{4K}\left(1+o_{t}\left(1\right)\right)\]
 where $o_{t}\left(1\right)$ is given in (\ref{eq:o_lb_ngL1WSC-1}),
then \begin{equation}
\underset{\left(m,N\right)\rightarrow\infty}{\lim}\Pr\left(\frac{1}{m}\left\Vert \mathbf{H}\mathbf{b}\right\Vert _{1}\le\delta\right)=0\label{eq:ngL1-delta}\end{equation}
 for all $\mathbf{b}\in B_{2t}^{N}$ and $\left\Vert \mathbf{b}\right\Vert _{0}\le2K$. 
\end{thm}
\begin{proof}
Similarly as for the corresponding proof for general $l_{1}$-WSCs,
we need a tight upper bound on the moment generation function of the
random variable \[
\left|\sum_{j=1}^{k}b_{j}\left|A_{i,j}\right|\right|.\]
 For this purpose, we resort to the use of the Central Limit Theorem.
We first approximate the distribution of $\sum_{j=1}^{k}b_{j}\left|A_{i,j}\right|$
by a Gaussian distribution. Then, we uniformly upper bound the approximation
error according to the Berry-Esseen Theorem (see~\cite{Book_Feller_II}
and Appendix~\ref{sub:BE-Theorem} for an overview of this theory).
Based on this approximation, we obtain an upper bound on the moment
generating function, with leading term $\left(\log k\right)/\sqrt{k}$
(see Equation~(\ref{eq:MGF_ub}) for details).

To simplify the notation, for a $\mathbf{b}\in B_{2t}^{N}$ with $\left\Vert \mathbf{b}_{0}\right\Vert =k$,
let \[
Y_{\mathbf{b},k}=\sum_{j=1}^{N}\frac{\sqrt{2\pi}}{2}\left|A_{j}\right|b_{j},\]
 where $A_{j}$s are standard Gaussian random variables. Then, \begin{align*}
 & \Pr\left(\frac{1}{m}\left\Vert \mathbf{H}\mathbf{b}\right\Vert _{1}\le\delta,\;\left\Vert \mathbf{b}\right\Vert _{0}=k\right)\\
 & \le\exp\left\{ m\left(\alpha\delta+\log\mathrm{E}\left[e^{-\alpha\left|Y_{\mathbf{b},k}\right|}\right]\right)\right\} ,\end{align*}
 where the inequality holds for all $\alpha>0$. Now, we fix $\alpha$
and upper bound the moment generating function as follows. Note that
\begin{align}
 & \mathrm{E}\left[e^{-\alpha\left|Y_{\mathbf{b},k}\right|}\right]\nonumber \\
 & =\mathrm{E}\left[e^{-\alpha\left|Y_{\mathbf{b},k}\right|},\;\left|Y_{\mathbf{b},k}\right|\ge\frac{1}{\alpha}\log\sqrt{k}\right]\label{eq:MGF_part1}\\
 & \quad+\mathrm{E}\left[e^{-\alpha\left|Y_{\mathbf{b},k}\right|},\;\left|Y_{\mathbf{b},k}\right|<\frac{1}{\alpha}\log\sqrt{k}\right].\label{eq:MGF_part2}\end{align}
 The first term (\ref{eq:MGF_part1}) is upper bounded by\begin{align}
 & \mathrm{E}\left[e^{-\alpha\frac{1}{\alpha}\log\sqrt{k}},\;\left|Y_{\mathbf{b},k}\right|\ge\frac{1}{\alpha}\log\sqrt{k}\right]\nonumber \\
 & \le\mathrm{E}\left[\frac{1}{\sqrt{k}},\;\left|Y_{\mathbf{b},k}\right|\ge\frac{1}{\alpha}\log\sqrt{k}\right]\nonumber \\
 & \le\frac{1}{\sqrt{k}}\Pr\left(\left|Y_{\mathbf{b},k}\right|\ge\frac{1}{\alpha}\log\sqrt{k}\right)\nonumber \\
 & \le\frac{1}{\sqrt{k}}.\label{eq:MGF_part1_ub}\end{align}
 In order to upper bound the second term in Equation~(\ref{eq:MGF_part2}),
we apply Lemma~\ref{lem:prob-scaled-by-sqrt(k)} from the Appendix,
proved using the Central Limit Theorem and the Berry-Esseen result:\begin{align}
 & \mathrm{E}\left[1,\;\left|Y_{\mathbf{b},k}\right|<\frac{1}{\alpha}\log\sqrt{k}\right]\nonumber \\
 & =\Pr\left(\left|Y_{\mathbf{b},k}\right|<\frac{1}{\alpha}\log\sqrt{k}\right)\nonumber \\
 & =\Pr\left(\left|\sum_{j=1}^{k}b_{j}\left|A_{j}\right|\right|<\frac{2}{\sqrt{2\pi}}\frac{1}{\alpha}\log\sqrt{k}\right)\nonumber \\
 & \le\frac{2}{\sqrt{2\pi}}\frac{1}{\alpha\pi}\frac{\log\sqrt{k}}{\sqrt{k}}+12\frac{t^{3}}{\sqrt{k}}\mathrm{E}\left[\left|A\right|^{3}\right]\nonumber \\
 & =\frac{1}{\sqrt{k}}\frac{1}{\sqrt{2\pi}}\left(\frac{\log k}{\alpha\pi}+48t^{3}\right).\label{eq:MGF_part2_ub}\end{align}
 Combining the upper bounds in~(\ref{eq:MGF_part1_ub}) and~(\ref{eq:MGF_part2_ub})
shows that \begin{align}
\mathrm{E}\left[e^{-\alpha\left|Y_{\mathbf{b},k}\right|}\right] & \le\frac{1}{\sqrt{k}}\left(1+\frac{1}{\sqrt{2\pi}}\left(\frac{\log k}{\alpha\pi}+48t^{3}\right)\right)\nonumber \\
 & \le\frac{1}{\sqrt{k}}\left(1+\frac{\log k}{4\alpha}+24t^{3}\right).\label{eq:MGF_ub}\end{align}
 Next, set $\alpha=1/\delta$. Then \begin{align*}
 & \Pr\left(\frac{1}{m}\left\Vert \mathbf{H}\mathbf{b}\right\Vert _{1}\le\delta,\;\left\Vert \mathbf{b}\right\Vert _{0}=k\right)\\
 & \le\exp\left\{ -m\left(\frac{1}{2}\log k\right)\left(1+o_{t,k}\left(1\right)\right)\right\} ,\end{align*}
 where\[
o_{t,k}\left(1\right)=-\frac{2+2\log\left(1+\frac{\log k}{4}+24t^{3}\right)}{\log k}.\]

Now we choose a $k_{0}\in\mathbb{Z}^{+}$ such that for all $k\ge k_{0}$,
\[
\frac{\log k}{2k}\left(1+o_{t,k}\left(1\right)\right)>0.\]
 It is straightforward to verify that $k_{0}$ is well defined. Consider
the case when $1\le k\le k_{0}$. It can be verified that \[
\sum_{j=1}^{k}b_{j}H_{i,j}=\frac{\sqrt{2\pi}}{2}\sum_{j=1}^{k}b_{j}\left|A_{i,j}\right|\]
 is Subgaussian and that \[
\mathrm{E}\left[\left|\sum_{j=1}^{k}b_{j}H_{i,j}\right|\right]\ge1\]
 for all $\mathbf{b}\in B_{2t}^{N}$ such that $\left\Vert \mathbf{b}\right\Vert _{0}=k$.
By applying the large deviations result for Subgaussian random variables,
as stated in Theorem \ref{thm:subgaussian-large-deviation}, and the
union bound, it can be proved that there exists a $c_{k}>0$ such
that \begin{align*}
 & \Pr\left(\bigcup_{\left\Vert \mathbf{b}\right\Vert _{0}=k}\left\{ \mathbf{H}:\;\frac{1}{m}\left\Vert \mathbf{H}\mathbf{b}\right\Vert _{1}\le\delta\right\} \right)\\
 & \le\exp\left\{ -mk\left(c_{k}-\frac{\log\left(4t\right)}{m}-\frac{\log N}{m}\right)\right\} \\
 & \rightarrow0.\end{align*}
 The above result holds whenever $m,N\rightarrow\infty$ simultaneously,
with \[
\underset{\left(m,N\right)\rightarrow\infty}{\lim}\frac{\log N}{m}<c_{k}.\]

Finally, let $K$ be sufficiently large so that \[
\frac{\log K}{4K}\left(1+o_{t,2K}\left(1\right)\right)\le\underset{k_{0}\le k\le2K}{\min}\frac{\log k}{2k}\left(1+o_{t,k}\left(1\right)\right),\]
 and \[
\frac{\log K}{4K}\left(1+o_{t,2K}\left(1\right)\right)\le\underset{1\le k\le k_{0}}{\min}c_{k}.\]
 Then \begin{align*}
 & \Pr\left(\bigcup_{\left\Vert \mathbf{b}\right\Vert _{0}\le2K}\;\left\{ \mathbf{H}:\;\frac{1}{m}\left\Vert \mathbf{H}\mathbf{b}\right\Vert _{1}\le\delta\right\} \right)\\
 & \le\sum_{k=1}^{2K}{N \choose k}\left(4t\right)^{k}\Pr\left(\frac{1}{m}\left\Vert \mathbf{H}\mathbf{b}\right\Vert _{1}\le\delta,\;\left\Vert \mathbf{b}\right\Vert _{0}=k\right)\\
 & \le\sum_{k=1}^{k_{0}}\exp\left\{ -mk\left(c_{k}-\frac{\log\left(4t\right)}{m}-\frac{\log N}{m}\right)\right\} \\
 & \quad+\sum_{k=k_{0}+1}^{2K}\exp\left\{ -mk\left(\frac{\log k}{2k}\left(1+o_{t,k}\left(1\right)\right)\right.\right.\\
 & \qquad\qquad\quad\left.\left.-\frac{\log\left(4t\right)}{m}-\frac{\log N}{m}\right)\right\} \\
 & \rightarrow0,\end{align*}
 as $m,N\rightarrow\infty$ with \[
\underset{\left(m,N\right)\rightarrow\infty}{\lim}\frac{\log N}{m}<\frac{\log K}{4K}\left(1+o_{t}\left(1\right)\right),\]
 where \begin{equation}
o_{t}\left(1\right)=-\frac{2+2\log\left(1+\frac{\log2K}{4}+24t^{3}\right)}{\log\left(2K\right)}.\label{eq:o_lb_ngL1WSC-1}\end{equation}

\end{proof}
\vspace{0.05in}

Based on Theorem~\ref{thm:L1-superposition-ngWSC-1}, we can characterize
the rate region in which any two distinct superpositions are sufficiently
separated in the $l_{1}$ space.

\begin{thm}
\label{thm:L1-superposition-ngWSC-2}Define $\mathbf{A}$ and $\mathbf{H}$
as in Theorem \ref{thm:L1-superposition-ngWSC-1}. For a given $\mathbf{H},$
define the diagonal matrix \[
\mathbf{\Lambda_{H}}=\left[\begin{array}{ccc}
m/\left\Vert \mathbf{h}_{1}\right\Vert _{1} &  & 0\\
 & \ddots\\
0 &  & m/\left\Vert \mathbf{h}_{N}\right\Vert _{1}\end{array}\right].\]
 Also, for $d\in\left(0,1\right)$, choose a $\delta\in\left(0,\frac{1}{2}\right)$
such that $d\left(1+\delta\right)<1$. Define the set $E_{1}$ as
in (\ref{eq:Event1-l1}). Provided that $K$ is sufficiently large,
if \[
\underset{\left(m,N\right)\rightarrow\infty}{\lim}\frac{\log N}{m}<\frac{\log K}{4K}\left(1+o_{t}\left(1\right)\right)\]
 where \begin{equation}
o_{t}\left(1\right)=-\frac{2+2\log\left(1+\frac{\log2K}{4}+648t^{3}\right)}{\log\left(2K\right)},\label{eq:o_lb_ngL1WSC-2}\end{equation}
 then it holds \[
\underset{\left(m,N\right)\rightarrow\infty}{\lim}\Pr\left(\frac{1}{m}\left\Vert \mathbf{H}\mathbf{\Lambda_{H}}\left(\mathbf{b}_{1}-\mathbf{b}_{2}\right)\right\Vert _{1}\le d,\; E_{1}\right)=0,\]
 for all pairs of $\mathbf{b}_{1},\mathbf{b}_{2}\in\mathcal{B}_{K}$
such that $\mathbf{b}_{1}\ne\mathbf{b}_{2}$. 
\end{thm}
\begin{proof}
The proof is very similar to that of Theorem~\ref{thm:L1-superposition-ngWSC-1}.
The only difference is the following. Let $\mathbf{b}^{\prime}=\mathbf{\Lambda_{H}}\left(\mathbf{b}_{1}-\mathbf{b}_{2}\right)$.
Since \[
\frac{1}{2}\le1-\delta\le\frac{1}{m}\left\Vert \mathbf{h}_{j}\right\Vert _{1}\le1+\delta\le\frac{3}{2},\]
 all the nonzero entries of $\left(1+\delta\right)\mathbf{b}^{\prime}$
on the set $E_{1}$ satisfy the following inequality \[
1\le\left|\left(1+\delta\right)b_{i}^{\prime}\right|\le3t.\]
 As a result, we have the higher order term $o_{t}\left(1\right)$
as given in Equation~(\ref{eq:o_lb_ngL1WSC-2}). 
\end{proof}

\section{Conclusions}

\label{sec:Conclusion}

We introduced a new family of codes over the reals, termed weighted
superimposed codes. Weighted superimposed codes can be applied to
all problems in which one seeks to robustly distinguish between bounded
integer valued linear combinations of codewords that obey predefined
norm and sign constraints. As such, they can be seen as a special
instant of compressed sensing schemes in which the sparse sensing
vectors contain entries from a symmetric, bounded set of integers.
We characterized the achievable rate regions of three classes of weighted
superimposed codes, for which the codewords obey $l_{2}$, $l_{1}$,
and non-negativity constraints.

\appendix

\subsection{\label{sub:SubGaussian-stuff}Subgaussian Random Variables}

\begin{definitn}
[The Subgaussian and Subexponential distributions]\label{def:subgaussian}
A random variable $X$ is said to be Subgaussian if there exist positive
constants $c_{1}$ and $c_{2}$ such that \[
\Pr\left(\left|X\right|>x\right)\le c_{1}e^{-c_{2}x^{2}}\quad\forall x>0.\]
 It is Subexponential if there exist positive constants $c_{1}$ and
$c_{2}$ such that \[
\Pr\left(\left|X\right|>x\right)\le c_{1}e^{-c_{2}x}\quad\forall x>0.\]

\end{definitn}
\vspace{0.05in}

\begin{lemma}
[Moment Generating Function]\label{lem:subgaussian-mgf}Let $X$
be a zero-mean random variable. Then, the following two statements
are equivalent. 
\begin{enumerate}
\item $X$ is Subgaussian. 
\item $\exists c$ such that $\mathrm{E}\left[e^{\alpha X}\right]\le e^{c\alpha^{2}}$,
$\forall\alpha\ge0$. 
\end{enumerate}
\end{lemma}
\vspace{0.05in}

\begin{thm}
\label{thm:subgaussian-large-deviation}Let $X_{1},\cdots,X_{n}$
be independent Subgaussian random variables with zero mean. For any
given $a_{1},\cdots,a_{n}\in\mathbb{R}$, $\sum_{k}a_{k}X_{k}$ is
a Subgaussian random variable. Furthermore, there exist positive constants
$c_{1}$ and $c_{2}$ such that \[
\Pr\left(\left|\sum_{k}a_{k}X_{k}\right|>x\right)\le c_{1}e^{-c_{2}x^{2}/\left\Vert \mathbf{a}\right\Vert _{2}^{2}},\quad\forall x>0,\]
 where $\left\Vert \mathbf{a}\right\Vert _{2}^{2}=\sum_{k}a_{k}^{2}$. 
\end{thm}
\begin{proof}
See \cite[Lecture 5, Theorem 5 and Corollary 6]{Vershynin_LectureNotes_nonasymptotic_RMT}. 
\end{proof}
\vspace{0.05in}

We prove next a result that asserts that translating a Subgaussian
random variable produces another Subgaussian random variable. 

\begin{prop}
\label{pro:subgaussian-shift}Let $X$ be a Subgaussian random variable.
For any given $a\in\mathbb{R}$, $Y=X+a$ is a Subgaussian random
variable as well. 
\end{prop}
\begin{proof}
It can be verified that for any $y\in\mathbb{R}$, \[
\left(y-a\right)^{2}\le\frac{1}{2}y^{2}-a^{2},\]
 and \[
\left(y+a\right)^{2}\le\frac{1}{2}y^{2}-a^{2}.\]

Now for $y>\left|a\right|$, \begin{align}
\Pr\left(\left|Y\right|>y\right) & =\Pr\left(X+a>y\right)+\Pr\left(X+a<-y\right)\nonumber \\
 & \le\Pr\left(X>y-a\right)+\Pr\left(X<-y-a\right).\label{eq:sgshift1}\end{align}
 When $a>0$, \begin{align}
(\ref{eq:sgshift1}) & \le\Pr\left(\left|X\right|>y-a\right)\nonumber \\
 & \le c_{1}e^{-c_{2}\left(y-a\right)^{2}}\nonumber \\
 & \le c_{1}c^{c_{2}a^{2}}e^{-c_{2}y^{2}/2}.\label{eq:sgshift2}\end{align}
 When $a\le0$, \begin{align}
(\ref{eq:sgshift1}) & \le\Pr\left(\left|X\right|>y+a\right)\nonumber \\
 & \le c_{1}e^{-c_{2}\left(y+a\right)^{2}}\nonumber \\
 & \le c_{1}c^{c_{2}a^{2}}e^{-c_{2}y^{2}/2}.\label{eq:sgshift3}\end{align}
 Combining Equations~(\ref{eq:sgshift2}) and (\ref{eq:sgshift3}),
one can show that \[
\Pr\left(\left|Y\right|>y\right)\le c_{1}c^{c_{2}a^{2}}e^{-c_{2}y^{2}/2},\;\forall y>\left|a\right|.\]
 On the other hand,\[
\Pr\left(\left|Y\right|\le y\right)\le1\le e^{c_{2}a^{2}/2}e^{-c_{2}y^{2}/2},\;\forall y\le\left|a\right|.\]

Let $c_{3}=\max\left(c_{1}e^{c_{2}a^{2}},e^{c_{2}a^{2}/2}\right)$
and $c_{4}=c_{2}/2$. Then \[
\Pr\left(\left|Y\right|>y\right)\le c_{3}e^{-c_{4}y^{2}}.\]
 This proves the claimed result. 
\end{proof}

\subsection{\label{sub:BE-Theorem}The Berry-Esseen Theorem and Its Consequence}

The Central Limit Theorem (CLT) states that under certain conditions,
an appropriately normalized sum of independent random variables converges
weakly to the standard Gaussian distribution. The Berry-Esseen theorem
quantifies the rate at which this convergence takes place.

\begin{thm}
[The Berry-Esseen Theorem]\label{thm:Berry-Esseen}Let $X_{1},X_{2},\ldots,X_{k}$
be independent random variables such that $\mathrm{E}\left[X_{i}\right]=0$,
$\mathrm{E}\left[X_{i}^{2}\right]=\sigma_{i}^{2}$, $\mathrm{E}\left[\left|X_{i}^{3}\right|\right]=\rho_{i}$.
Also, let \[
s_{k}^{2}=\sigma_{1}^{2}+\cdots+\sigma_{k}^{2},\]
 and \[
r_{k}=\rho_{1}+\cdots+\rho_{k}.\]
 Denote by $F_{k}$ the cumulative distribution function of the normalized
sum $\left(X_{1}+\cdots+X_{k}\right)/s_{k}$, and by $\mathscr{N}$
the standard Gaussian distribution. Then for all $x$ and $k$, \[
\left|F_{k}\left(x\right)-\mathscr{N}\left(x\right)\right|\le6\frac{r_{k}}{s_{k}^{3}}.\]

\vspace{0.05in}

The Berry-Esseen theorem is used in the proof of the lower bound for
the achievable rate region of nonnegative $l_{1}$-WSCs. In the proof,
one need to identify a tight bound on the probability of a weighted
sum of nonnegative random variables. The probability of this sum lying
in a given interval can be estimated by the Berry-Esseen, as summarized
in the following lemma. 
\end{thm}
\begin{lemma}
\label{lem:prob-scaled-by-sqrt(k)} Assume that $\mathbf{b}\in B_{t}^{k}$
is such that $\left\Vert \mathbf{b}\right\Vert _{0}=k$, let $X_{1},X_{2},\cdots,X_{k}$
be independent standard Gaussian random variables. For a given positive
constant $c>0$, one has \[
\Pr\left(\left|\sum_{j=1}^{k}b_{j}\left|X_{j}\right|\right|<c\log\sqrt{k}\right)\le\frac{c}{\pi}\frac{\log\sqrt{k}}{\sqrt{k}}+12\rho\frac{t^{3}}{\sqrt{k}},\]
 where $\rho:=\mathrm{E}\left[\left|X\right|^{3}\right]$. 
\end{lemma}
\begin{proof}
This lemma is proved by applying the Berry-Essen theorem. Note that
the $b_{j}\left|X_{j}\right|$'s are independent random variables.
Their sum $\sum_{j=1}^{k}b_{j}\left|X_{j}\right|$ can be approximated
by a Gaussian random variable with properly chosen mean and variance,
according to the Central Limit Theorem. In the proof, we first use
the Gaussian approximation to estimate the probability \[
\Pr\left(\left|\sum_{j=1}^{k}b_{j}\left|X_{j}\right|\right|<c\log\sqrt{k}\right).\]
 Then we subsequently employ the Berry-Essen theorem to upper bound
the approximation error.

To simplify notation, let \[
Y_{\mathbf{b},k}=\sum_{j=1}^{k}b_{j}\left|X_{j}\right|,\]
 and let $\mathscr{N}\left(x\right)$ denote, as before, the standard
Gaussian distribution. Then,\begin{align*}
 & \Pr\left(\left|Y_{\mathbf{b},k}\right|<c\log\sqrt{k}\right)\\
 & \le\Pr\left(\frac{Y_{\mathbf{b},k}}{\sqrt{\sum_{j}b_{j}^{2}}}\in\left(-\frac{c\log\sqrt{k}}{\sqrt{\sum_{j}b_{j}^{2}}},\frac{c\log\sqrt{k}}{\sqrt{\sum_{j}b_{j}^{2}}}\right)\right)\\
 & \le\Pr\left(\frac{Y_{\mathbf{b},k}}{\sqrt{\sum_{j}b_{j}^{2}}}\in\left(-\frac{c\log\sqrt{k}}{\sqrt{k}},\frac{c\log\sqrt{k}}{\sqrt{k}}\right)\right)\\
 & \le\Pr\left(\frac{Y_{\mathbf{b},k}}{\left\Vert \mathbf{b}\right\Vert _{2}}\le\frac{c\log\sqrt{k}}{\sqrt{k}}\right)-\Pr\left(\frac{Y_{\mathbf{b},k}}{\left\Vert \mathbf{b}\right\Vert _{2}}\le\frac{c\log\sqrt{k}}{\sqrt{k}}\right),\end{align*}
 where in the second inequality we used the fact that $b_{j}\geq1$,
so that $\sum_{j=1}^{k}\, b_{j}^{2}\geq k$.

According to Theorem \ref{thm:Berry-Esseen}, for all $x\in\mathbb{R}$
and all $k$,\begin{align*}
 & \left|\Pr\left(\frac{Y_{\mathbf{b},k}}{\left\Vert \mathbf{b}\right\Vert _{2}}\le x\right)-\mathscr{N}\left(x\right)\right|\\
 & \le6\rho\,\frac{\sum_{j=1}^{k}\left|b_{j}\right|^{3}}{\left(\sum_{j=1}^{k}\left|b_{j}\right|^{2}\right)^{3/2}}\\
 & \le\frac{6k\rho t^{3}}{k^{3/2}}=\frac{6\rho t^{3}}{\sqrt{k}},\end{align*}
 since $\sum_{j=1}^{k}\,|b_{j}|^{3}\leq k\, t^{3}$, and $\sum_{j=1}^{k}\,|b_{j}|^{2}\geq k$.

Thus,\begin{align*}
 & \Pr\left(\left|Y_{\mathbf{b},k}\right|<c\log\sqrt{k}\right)\\
 & \le\mathscr{N}\left(\frac{c\log\sqrt{k}}{\sqrt{k}}\right)+\frac{6\rho t^{3}}{\sqrt{k}}\\
 & \quad-\mathscr{N}\left(-\frac{c\log\sqrt{k}}{\sqrt{k}}\right)+\frac{6\rho t^{3}}{\sqrt{k}}\\
 & \le\frac{2}{2\pi}\frac{c\log\sqrt{k}}{\sqrt{k}}+\frac{12\rho t^{3}}{\sqrt{k}},\end{align*}
 which completes the proof of the claimed result. 
\end{proof}
\bibliographystyle{IEEEtran} \bibliographystyle{IEEEtran}
\bibliography{CompressedSensing,SuperImposedCodes}

\begin{thebibliography}{10}
\providecommand{\url}[1]{#1}
\csname url@samestyle\endcsname
\providecommand{\newblock}{\relax}
\providecommand{\bibinfo}[2]{#2}
\providecommand{\BIBentrySTDinterwordspacing}{\spaceskip=0pt\relax}
\providecommand{\BIBentryALTinterwordstretchfactor}{4}
\providecommand{\BIBentryALTinterwordspacing}{\spaceskip=\fontdimen2\font plus
\BIBentryALTinterwordstretchfactor\fontdimen3\font minus
  \fontdimen4\font\relax}
\providecommand{\BIBforeignlanguage}[2]{{%
\expandafter\ifx\csname l@#1\endcsname\relax
\typeout{** WARNING: IEEEtran.bst: No hyphenation pattern has been}%
\typeout{** loaded for the language `#1'. Using the pattern for}%
\typeout{** the default language instead.}%
\else
\language=\csname l@#1\endcsname
\fi
#2}}
\providecommand{\BIBdecl}{\relax}
\BIBdecl

\bibitem{Kautz_IT1964_Superimposed}
W.~Kautz and R.~Singleton, ``Nonrandom binary superimposed codes,'' \emph{IEEE
  Trans. Inform. Theory}, vol.~10, no.~4, pp. 363--377, 1964.

\bibitem{Ericson_IT1988_SuperImposed_Codes_Rn}
T.~Ericson and L.~Gy{\"o}rfi, ``Superimposed codes in {${\bf R}\sp n$},''
  \emph{IEEE Trans. Inform. Theory}, vol.~34, no.~4, pp. 877--880, 1988.

\bibitem{Furedi_IT1999_ub_rate_superimposed_codes}
Z.~F{\"u}redi and M.~Ruszink{\'o}, ``An improved upper bound of the rate of
  {E}uclidean superimposed codes,'' \emph{IEEE Trans. Inform. Theory}, vol.~45,
  no.~2, pp. 799--802, 1999.

\bibitem{Donoho_IT2006_CompressedSensing}
D.~Donoho, ``Compressed sensing,'' \emph{IEEE Trans. Inform. Theory}, vol.~52,
  no.~4, pp. 1289--1306, 2006.

\bibitem{Candes_Tao_IT2006_Robust_Uncertainty_Principles}
E.~Cand{\`e}s, J.~Romberg, and T.~Tao, ``Robust uncertainty principles: exact
  signal reconstruction from highly incomplete frequency information,''
  \emph{IEEE Trans. Inform. Theory}, vol.~52, no.~2, pp. 489--509, 2006.

\bibitem{Candes_Tao_ApplMath2006_Stable_Signal_Recovery}
E.~J. Cand{\`e}s, J.~K. Romberg, and T.~Tao, ``Stable signal recovery from
  incomplete and inaccurate measurements,'' \emph{Comm. Pure Appl. Math.},
  vol.~59, no.~8, pp. 1207--1223, 2006.

\bibitem{Candes_Tao_IT2006_Near_Optimal_Signal_Recovery}
E.~J. Cand{\`e}s and T.~Tao, ``Near-optimal signal recovery from random
  projections: Universal encoding strategies?'' \emph{IEEE Trans. Inform.
  Theory}, vol.~52, no.~12, pp. 5406--5425, 2006.

\bibitem{Cormode_2003}
G.~Cormode and S.~Muthukrishnan, ``What's hot and what's not: Tracking most
  frequent items dynamically,'' \emph{IEEE Trans. Inform. Theory}, vol.~50,
  no.~10, pp. 2231--2242, 2004.

\bibitem{Dai2008_ITW_Weighted_Euclidean_Superimposed_Code}
W.~Dai and O.~Milenkovic, ``Weighted euclidean superimposed codes for integer
  compressed sensing,'' in \emph{IEEE Information Theory Workshop (ITW)}, 2008,
  submitted.

\bibitem{Dai2008_ITA_constrained_CS}
------, ``Constrained compressed sensing via superimposed coding,'' in
  \emph{Information Theory and Applications Workshop}, San Diego, CA, invited
  talk, Jan. 2008.

\bibitem{Dai2008_CISS_sparse_superimposed_codes}
------, ``Sparse weighted euclidean superimposed coding for integer compressed
  sensing,'' in \emph{Conference on Infomation Sciences and Systems (CISS)},
  2008, submitted.

\bibitem{Sheikh2007_CAMSAP_microarray}
M.~Sheikh, O.~Milenkovic, and R.~Baraniuk, ``Designing compressive sensing
  {DNA} microarrays,'' \emph{Proceedings of the IEEE Workshop on Computational
  Advances in Multi-Sensor Adaptive Processing (CAMSAP), St. Thomas, U.S.
  Virgin Islands}, Dec. 2007.

\bibitem{Dai2008_BIBM_microarray}
W.~Dai, M.~Sheikh, O.~Milenkovic, and R.~Baraniuk, ``Probe designs for
  compressed sensing microarrays,'' in \emph{IEEE International Conference on
  Bioinformatics and Biomedicine, Philadelphia, PA}, submitted, 2008.

\bibitem{Book_Feller_II}
W.~Feller, \emph{An Introduction to Probability Theory and Its Applications,
  Volume 2}, 2nd~ed.\hskip 1em plus 0.5em minus 0.4em\relax Wiley, 1971.

\bibitem{Vershynin_LectureNotes_nonasymptotic_RMT}
R.~Vershynin, \emph{Non-asymptotic Random Matrix Theory (Lecture Notes)}, 2007.

\end{thebibliography}

\end{document}